\documentclass[journal]{IEEEtran}
\usepackage{amsmath,amssymb,amsfonts}
\usepackage{graphicx}
\usepackage{multicol}
\usepackage{footmisc}
\usepackage{subcaption}
\usepackage{diagbox} 
\usepackage{cite}
\usepackage{hyperref}
\hypersetup{hidelinks=true}
\usepackage{textcomp}
\usepackage{xcolor}

\usepackage{amsthm}

\newtheoremstyle{ieeestyle}
  {.5em plus .2em minus .2em}
  {.5em plus .2em minus .2em}
  {\itshape}
  {}
  {\bfseries}
  {:}
  { }
  {\thmname{#1}\ \thmnumber{#2}\if\relax\detokenize{#3}\relax\else\ (#3)\fi}

\newtheoremstyle{ieeedef}
  {.5em plus .2em minus .2em}
  {.5em plus .2em minus .2em}
  {\normalfont}
  {}
  {\bfseries}
  {:}
  { }
  {\thmname{#1}\ \thmnumber{#2}\ (\thmnote{#3})}

\newtheoremstyle{ieeeremark}
  {.5em plus .2em minus .2em}
  {.5em plus .2em minus .2em}
  {\normalfont}
  {}
  {\bfseries}
  {:}
  { }
  {\thmname{#1}\ \thmnumber{#2}\if\relax\detokenize{#3}\relax\else\ (#3)\fi}

\theoremstyle{ieeestyle}
\newtheorem{theorem}{Theorem}
\newtheorem{lemma}{Lemma}

\theoremstyle{ieeedef}
\newtheorem{definition}{Definition}

\theoremstyle{ieeeremark}
\newtheorem{remark}{Remark}

\usepackage{algorithm}
\usepackage{algorithmicx}
\usepackage{algpseudocode}

\def\BibTeX{{\rm B\kern-.05em{\sc i\kern-.025em b}\kern-.08em
    T\kern-.1667em\lower.7ex\hbox{E}\kern-.125emX}}

\begin{document}
\title{CT-ESKF: A General Framework of Covariance Transformation-Based Error-State Kalman Filter}
\author{Jiale Han, Wei Ouyang, \IEEEmembership{Member, IEEE}, Maoran Zhu, \IEEEmembership{Member, IEEE}, and Yuanxin Wu, \IEEEmembership{Senior Member, IEEE}
\thanks{This work was supported in part by National Key R\&D Program under Grant 2022YFB3903802 and in part by National Natural Science Foundation under Grant 62303310, Grant 62273228, and Grant 62403315. (Corresponding author: Yuanxin Wu.)}
\thanks{Jiale Han, Maoran Zhu, and Yuanxin Wu are with the Shanghai Key Lab of Navigation and Location Service, School of Automatic and Intelligent Sensing, Shanghai Jiao Tong University, Shanghai 200240, China (e-mails: jhan04573@gmail.com, zhumaoran@sjtu.edu.cn, yuanx\_wu@hotmail.com). }
\thanks{Wei Ouyang is with the College of Surveying and Geo-Informatics, Tongji University, Shanghai 200092, China (e-mail: ywoulife@tongji.edu.cn).}
}

\maketitle

\begin{abstract}
Invariant extended Kalman filter (InEKF) possesses excellent trajectory-independent property and better consistency compared to conventional extended Kalman filter (EKF). However, when applied to scenarios involving both global-frame and body-frame observations, InEKF may fail to preserve its trajectory-independent property. This work introduces the concept of equivalence between error states and covariance matrices among different error-state Kalman filters, and shows that although InEKF exhibits trajectory independence, its covariance propagation is actually equivalent to EKF. A covariance transformation-based error-state Kalman filter (CT-ESKF) framework is proposed that unifies various error-state Kalman filtering algorithms. The framework gives birth to novel filtering algorithms that demonstrate improved performance in integrated navigation systems that incorporate both global and body-frame observations. Experimental results show that the EKF with covariance transformation outperforms both InEKF and original EKF in a representative INS/GNSS/Odometer integrated navigation system.
\end{abstract}


\section{Introduction}
\label{sec:introduction}
\IEEEPARstart{T}{he} extended Kalman filter (EKF) has been widely applied to various nonlinear state estimation problems, including integrated navigation systems and robotic platforms, and becomes a standard algorithm in engineering applications~\cite{kalman1960new,schmidt1981kalman, niu2019data, zhu2021f2imu, liu2020simple}. When applied to attitude estimation problems, the EKF or the more appropriately termed error-state Kalman filter (ESKF) typically represents attitude on the special orthogonal group ($\mathbb{SO}(3)$) or quaternion, while other states such as velocity and position are defined in vector spaces~\cite{markley2003attitude, groves2015principles}. However, the error representation does not fully account for the coupling between attitude, velocity, and position errors, nor does it adequately reflect the intrinsic geometric structure of the estimation problem, which hinders to achieve optimal estimation performance~\cite{bonnabel2008symmetry, andrle2015attitude, wang2019consistent, chang2019strapdown}. As a result, the EKF often suffers from inconsistency and apparent drift in applications without absolute observations, such as visual-inertial odometry (VIO)~\cite{sun2018robust, clement2015battle}, LiDAR-inertial odometry (LIO)~\cite{he2023point} and simultaneous localization and mapping (SLAM)~\cite{barrau2015ekf}.

In contrast, the invariant extended Kalman filter (InEKF) represents the attitude, velocity, and position on the group of double direct isometries ($\mathbb{SE}_2(3)$), and replaces the additive errors of velocity and position used in traditional EKF with multiplicative errors defined on matrix Lie groups, thereby fully accounting for the geometric structure of the estimation problem~\cite{barrau2016invariant}. Due to its group-affine property, InEKF exhibits strong consistency and local convergence~\cite{barrau2016invariant, chang2022log} and has been widely applied in various fields, including VIO~\cite{wu2017invariant, heo2018consistent, hua2023piekf, du2024sp}, LIO~\cite{zhang2024si, brossard2021associating}, initial alignment~\cite{chang2021strapdown, cui2024enhanced}, and integrated navigation~\cite{hwang2022novel, cui2021lie, tang2022invariant, luo2023filter, luo2023matrix, du2024novel}. InEKF possesses the trajectory-independent property, i.e., its system and observation matrices are mostly independent of the estimated trajectory, whereas in EKF the system matrix is trajectory-dependent. It is commonly believed that InEKF diminishes linearization errors caused by erroneous state and thus provides superior consistency in covariance propagation ~\cite{zhang2024si, luo2023filter}, particularly under large initialization errors.

InEKF can be categorized as left-invariant (L-InEKF) or right-invariant (R-InEKF) depending on how the group error is defined~\cite{barrau2016invariant, hartley2019contact}. Experiences indicate that observations in the global frame are more suitable for left-invariant error, while observations in the body frame prefer the right-invariant error~\cite{barrau2016invariant, hartley2019contact, chang2021inertial}. When L-InEKF is used in systems containing body-frame observations, or when R-InEKF is applied to systems with global-frame observations, inconsistency issues would typically arise and lead to degraded estimation accuracy in heterogeneous sensor configurations~\cite{barrau2016invariant, han2024covariance}, consisting of global/body-hybrid frame sensors, such as a sensor suite containing inertial measurement unit (IMU), LiDAR, odometer (ODO), camera, and global navigation satellite system (GNSS) receiver~\cite{cui2021lie, chiang2019assessment, zhang2021gnss}. To address this problem in InEKF, the concept of covariance switch has been employed~\cite{hartley2019contact, hwang2022novel, van2020invariant, Han2024}. Specifically, during the Kalman update, the current error state was transformed into an alternative error state, so as to ensure consistency between the process and observation models~\cite{hartley2019contact, Han2024}. In~\cite{hwang2022novel}, a federated InEKF was proposed using covariance switch to fuse state updates from two distinct observation models via a least-squares solution. Reference~\cite{van2020invariant} handled heterogeneous scenarios by using a similar approach in smoothing.

Theoretical analyses in our previous works~\cite{han2024covariance, Han2024} showed that covariance switch in InEKF significantly improves convergence speed in integrated navigation systems. Specifically, the work~\cite{Han2024} indicated that the idea of covariance switch can be directly extended to other error-state Kalman filters. Quite recently, Chen~\cite{chen2024visual} demonstrated the effect of covariance switch on consistency improvement in visual-inertial navigation systems. The work~\cite{hao2024consistency} used a similar idea in multi-robot cooperative localization and achieved promising localization performance. However, these studies lack a systematic investigation of the underlying mechanism by which covariance switch contributes to performance improvement.

Deepening our previous work~\cite{Han2024}, the current paper throws lights on an interesting interaction between geometryic error representation and covariance, namely, all ESKFs could be transformed to each other by appropriate covariance transformations. In summary, the main contributions of this paper are as follows:
\begin{enumerate}
    \item A novel covariance transformation-based error-state Kalman filter framework is proposed. The differences between covariance transformation and covariance switch are systematically investigated.
    
    \item The concept of equivalence between error states and covariance matrices of various ESKFs is introduced, and applied to study their covariance propagation.

    \item Theoretical analyses and real-world experiments show that incorporating covariance transformation into the EKF significantly enhances convergence speed under large initial attitude errors in multi-sensor fusion systems.
\end{enumerate}
\section{Fundamentals of ESKF}\label{sec:filter_Preliminary}
In practical state estimation problems, the system dynamics are typically continuous, whereas observations are obtained at discrete time instants. As a result, the estimation problem is generally modeled as a continuous-discrete nonlinear system~\cite{crassidis2004optimal}
\begin{equation}
    \dot{\mathbf{x}}(t) = \mathbf{f}(\mathbf{x}(t), \mathbf{u}(t), t) + \mathbf{G}(t)\mathbf{w}(t), \quad t \in [t_0, t_{end}],
\end{equation}
\begin{equation}
    \tilde{\mathbf{y}}(t_k) = \mathbf{h}(\mathbf{x}(t_k), t_k) + \mathbf{v}(t_k), \quad t_k \in [t_0, t_{end}],
\end{equation}
where $t_0$ and $t_{end}$ denote the initial and terminal time, respectively, and $t_k$ represents the time instants at which observations are acquired. In this paper, $(\hat{\cdot})$ and $(\tilde{\cdot})$ denote the estimated and measured values, respectively. The variable $\mathbf{x}(t)$ denotes the time-varying system state, $\mathbf{u}(t)$ is the system input, and $\mathbf{G}(t)$ is the noise driving matrix. The function $\mathbf{f}(\cdot)$ defines the system dynamics. The variable $\tilde{\mathbf{y}}(t_k)$ represents the observation, and $\mathbf{h}(\cdot)$ denotes the nonlinear observation function. The continuous process noise $\mathbf{w}(t)$ is defined with covariance $\mathbf{E}\{\mathbf{w}(t) \, \mathbf{w}(\tau)\} = \mathbf{Q}(t) \delta(t - \tau)$, where $\delta(\cdot)$ is the Dirac’s delta function. The observation noise $\mathbf{v}(t_k)$ is assumed to be zero-mean Gaussian noise with covariance $\mathbf{R}(t_k)$, where $\mathbf{R}(t_k) = \mathbf{E}\{\mathbf{v}(t_k) \, \mathbf{v}^T(t_k)\}$.

State estimation methods based on the Kalman filtering typically consist of two stages: the propagation step and the update step~\cite{kalman1960new, schmidt1981kalman}. The propagation step describes the evolution of the estimated system state and its associated uncertainty over time. The propagation model of the estimated state is given by~\cite{crassidis2004optimal}
\begin{equation}
    \dot{\hat{\mathbf{x}}}(t) = \mathbf{f}(\hat{\mathbf{x}}(t), \mathbf{u}(t), t), \quad t \in [t_0, t_{end}].
\end{equation}

EKF is a commonly used linearization-based approach, which assumes that the current estimated state is close to the true system state. The difference between the estimated state and the true state is defined as the error state,
\begin{equation}
    \label{eq:delta_x_definition}
    \delta \mathbf{x} = \hat{\mathbf{x}} - \mathbf{x},
\end{equation}
\begin{equation}
\begin{aligned}
    \delta \dot{\mathbf{x}}(t) &= \mathbf{f}(\hat{\mathbf{x}}(t), \mathbf{u}(t), t) - \mathbf{f}(\mathbf{x}(t), \mathbf{u}(t), t) - \mathbf{G}(t) \mathbf{w}(t)\\
                               &\approx \mathbf{F}(t) \delta \mathbf{x}(t) - \mathbf{G}(t) \mathbf{w}(t),
\end{aligned}
\end{equation}
where $\mathbf{F}(t)$ is the Jacobian of the system function $\mathbf{f}(\cdot)$ with respect to $\mathbf{x}$, evaluated at the estimated state $\hat{\mathbf{x}}(t)$
\begin{equation}
    \left. \mathbf{F}(t) = \frac{\partial \mathbf{f}}{\partial \mathbf{x}} \right |_{\hat{\mathbf{x}}(t), \mathbf{u}(t)}.
\end{equation}
The system uncertainty $\mathbf{P}$ is defined as
\begin{equation}
    \mathbf{P} = \mathbf{E}[\delta \mathbf{x}(t) \; \delta \mathbf{x}^T(t)],
\end{equation}
and its evolution is governed by
\begin{equation}
    \dot{\mathbf{P}}(t) = \mathbf{F}(t) \mathbf{P}(t) + \mathbf{P}(t) \mathbf{F}^T(t) + \mathbf{G}(t) \mathbf{Q}(t) \mathbf{G}^T(t), \; t \in [t_0, t_{end}].
\end{equation}
 When observations are available at time $t_k$, the Kalman update is performed using the innovation $\delta \mathbf{z}$ and observation matrix $\mathbf{H}$ is computed as follows
\begin{equation}
    \label{eq:inov_all}
    \delta \mathbf{z} = \hat{\mathbf{y}}(t_k) - \tilde{\mathbf{y}}(t_k),
\end{equation}
\begin{equation}
    \left. \mathbf{H}(t_k) = \frac{\partial \mathbf{h}}{\partial \mathbf{x}} \right|_{\hat{\mathbf{x}}(t_k), \mathbf{u}(t_k)}.
\end{equation}

The Kalman gain $\mathbf{K}$ for the continuous-discrete Kalman filtering algorithm is computed by
\begin{equation}
    \mathbf{K}(t_k) = \mathbf{P}^-(t_k) \mathbf{H}^T(t_k) \left[\mathbf{H}(t_k) \mathbf{P}^-(t_k) \mathbf{H}^T(t_k) + \mathbf{R}(t_k) \right]^{-1}.
\end{equation}
Based on $\mathbf{K}$, the error state and covariance are updated as
\begin{equation}
    \delta \mathbf{x}^+(t_k) = \delta \mathbf{x}^-(t_k) + \mathbf{K}(t_k) \delta \mathbf{z},
\end{equation}
\begin{equation}
    \mathbf{P}^+(t_k) = \mathbf{P}^-(t_k) - \mathbf{K}(t_k) \mathbf{H}(t_k) \mathbf{P}^-(t_k).
\end{equation}
In this paper, $(\cdot)^-$ and $(\cdot)^+$ denote the predicted and updated values, respectively.

The propagation of the system state $\mathbf{x}$, the error state $\delta \mathbf{x}$, and the covariance matrix $\mathbf{P}$ in an arbitrary interval $[t_s, t_f]$ is computed as follows~\cite{groves2015principles}
\begin{equation}
    \label{eq:x_prop}
    \hat{\mathbf{x}}^-(t_f) = \hat{\mathbf{x}}^+(t_s) + \int_{t_s}^{t_f} \mathbf{f}(\hat{\mathbf{x}}(t'), \mathbf{u}(t'), t') \, d t',
\end{equation}
\begin{equation}
    \delta \mathbf{x}^-(t_f) = \delta \mathbf{x}^+(t_s) + \int_{t_s}^{t_f} \mathbf{F}(t') \delta \mathbf{x}(t') - \mathbf{G}(t') \mathbf{w}(t')\, d t',
\end{equation}
\begin{equation}
\begin{aligned}
    \mathbf{P}^-(t_f) = &\mathbf{P}^+(t_s) + \int_{t_s}^{t_f} \mathbf{F}(t') \mathbf{P}(t') + \mathbf{P}(t') \mathbf{F}^T(t') \\
                        &+ \mathbf{G}(t') \mathbf{Q}(t') \mathbf{G}^T(t') \, d t'.
\end{aligned}
\end{equation}
\section{Relationship Between Error-State Representations} \label{sec:error_state_relation}
Consider the system state $\mathbf{x}$ associated with two different definitions of error states, denoted as $\boldsymbol{\xi}_a$ and $\boldsymbol{\xi}_b$. The corresponding linearized differential equations are given by
\begin{equation}
    \label{eq:diff_state_a}
    \dot{\boldsymbol{\xi}}_a = \mathbf{F}_a \boldsymbol{\xi}_a + \mathbf{G}_a \mathbf{w},
\end{equation}
\begin{equation}
    \label{eq:diff_state_b}
    \dot{\boldsymbol{\xi}}_b = \mathbf{F}_b \boldsymbol{\xi}_b + \mathbf{G}_b \mathbf{w},
\end{equation}
where $\mathbf{F}_a$ and $\mathbf{F}_b$ are the system matrices and $\mathbf{G}_a$ and $\mathbf{G}_b$ are the corresponding noise driving matrices. It is important to note that the process noise $\mathbf{w}$ remains the same up to first-order approximation, and only the noise driving matrices differ. Hereafter, ESKF(a) and ESKF(b) denote Kalman filtering algorithms with error states $\boldsymbol{\xi}_a$ and $\boldsymbol{\xi}_b$, respectively.

If an invertible matrix $\mathbf{A}(\hat{\mathbf{x}})$ exists between $\boldsymbol{\xi}_a$ and $\boldsymbol{\xi}_b$, i.e.,
\begin{equation}
    \label{eq:err_relation}
    \boldsymbol{\xi}_a = \mathbf{A}(\hat{\mathbf{x}}) \boldsymbol{\xi}_b,
\end{equation}
\begin{equation}
        \mathbf{A}(\hat{\mathbf{x}}) \boldsymbol{\xi}_b = \mathbf{A}(\mathbf{x} + \delta \mathbf{x}) \boldsymbol{\xi}_b \approx \mathbf{A}(\mathbf{x}) \boldsymbol{\xi}_b.
\end{equation}
\begin{remark}
    The matrix $\mathbf{A}$ may depend on either the estimated state $\hat{\mathbf{x}}$ or the true state $\mathbf{x}$, since higher-order error terms are typically neglected during first-order linearization.
\end{remark}

Taking the time derivative of~\eqref{eq:err_relation} yields
\begin{equation}
\begin{aligned}
    \label{eq:diff_relation1}
    \dot{\boldsymbol{\xi}}_a &= \dot{\mathbf{A}}(\hat{\mathbf{x}}) \boldsymbol{\xi}_b + \mathbf{A}(\hat{\mathbf{x}}) \dot{\boldsymbol{\xi}}_b \\
                             &= \dot{\mathbf{A}}(\hat{\mathbf{x}}) \boldsymbol{\xi}_b + \mathbf{A}(\hat{\mathbf{x}}) \left(\mathbf{F}_b \boldsymbol{\xi}_b + \mathbf{G}_b \mathbf{w}\right) \\
                             &= \left(\dot{\mathbf{A}}(\hat{\mathbf{x}}) + \mathbf{A}(\hat{\mathbf{x}}) \mathbf{F}_b \right) \boldsymbol{\xi}_b + \mathbf{A}(\hat{\mathbf{x}}) \mathbf{G}_b \mathbf{w}.
\end{aligned}
\end{equation}
Substituting \eqref{eq:err_relation} into \eqref{eq:diff_state_a} yields
\begin{equation}
    \label{eq:diff_relation2}
    \dot{\boldsymbol{\xi}}_a = \mathbf{F}_a \mathbf{A}(\hat{\mathbf{x}}) \boldsymbol{\xi}_b + \mathbf{G}_a \mathbf{w}.
\end{equation}
Combining \eqref{eq:diff_relation1} and \eqref{eq:diff_relation2} gives the relationships between the system matrices and noise driving matrices corresponding to different error states
\begin{equation}
    \label{eq:F_relation}
    \mathbf{F}_a \mathbf{A}(\hat{\mathbf{x}}) = \dot{\mathbf{A}}(\hat{\mathbf{x}}) + \mathbf{A}(\hat{\mathbf{x}}) \mathbf{F}_b,
\end{equation}
\begin{equation}
    \label{eq:G_relation}
    \mathbf{G}_a = \mathbf{A}(\hat{\mathbf{x}}) \mathbf{G}_b.
\end{equation}
The observation matrices $\mathbf{H}_a$ and $\mathbf{H}_b$, associated respectively with $\boldsymbol{\xi}_a$ and $\boldsymbol{\xi}_b$, satisfy
\begin{equation}
    \delta \mathbf{z} = \mathbf{H}_a \boldsymbol{\xi}_a^- + \mathbf{v}= \mathbf{H}_b \boldsymbol{\xi}_b^- + \mathbf{v},
\end{equation}
and the transformation between $\mathbf{H}_a$ and $\mathbf{H}_b$ is given by
\begin{equation}
    \label{eq:H_relation}
    \mathbf{H}_b = \mathbf{H}_a \mathbf{A}(\hat{\mathbf{x}}).
\end{equation}
\begin{remark}
    It is noteworthy that in both ESKF(a) and ESKF(b), if the predicted system states at time $t_k$ are identical, then for the same observation, the innovation $\delta \mathbf{z}$ is computed according to~\eqref{eq:inov_all} using the same predicted $\hat{\mathbf{y}}(t_k)$ and the actual observation $\tilde{\mathbf{y}}(t_k)$. This guarantees that the corresponding observation matrices satisfy the relation in~\eqref{eq:H_relation}. With this observation, the observation noise covariance $\mathbf{R}$ is the same across different error state definitions, representing only the uncertainty of the observations and being independent of the chosen error state.
\end{remark}
We will see in the sequel that the error states $\boldsymbol{\xi}_a$ and $\boldsymbol{\xi}_b$ discussed above are not restricted to specific forms; they can represent the traditional EKF error $\delta \mathbf{x}$ or any other linearized error forms, such as the left-invariant and right-invariant errors in InEKF.

For $\boldsymbol{\xi}_a$ and $\boldsymbol{\xi}_b$ and their associated covariance matrices $\mathbf{P}_a$ and $\mathbf{P}_b$, we next introduce the concept of equivalence of error states and covariance matrices, which was partly discussed in our previous work~\cite{Han2024, han2024covariance}.

\begin{definition}[Equivalent Error States]
    \label{def:error_equ}
    Two arbitrary error states $\boldsymbol{\xi}_a$ and $\boldsymbol{\xi}_b$ are said to be equivalent if they satisfy
    \begin{equation}
        \label{eq:error_state_switch}
        \boldsymbol{\xi}_a = \mathbf{A}(\hat{\mathbf{x}}) \boldsymbol{\xi}_b,
    \end{equation}
    which is called the equivalent error state switch.

    Specifically, if the current system state is the predicted state $\hat{\mathbf{x}}^-$, then the equivalence between two predicted error states $\boldsymbol{\xi}_a^-$ and $\boldsymbol{\xi}_b^-$ satisfy
    \begin{equation}
        \boldsymbol{\xi}_a^- = \mathbf{A}(\hat{\mathbf{x}}^-) \boldsymbol{\xi}_b^-.
    \end{equation}
    Similarly, if the current system state is the updated state $\hat{\mathbf{x}}^+$, then the equivalent two updated error states $\boldsymbol{\xi}_a^+$ and $\boldsymbol{\xi}_b^+$ satisfy
    \begin{equation}
        \boldsymbol{\xi}_a^+ = \mathbf{A}(\hat{\mathbf{x}}^+) \boldsymbol{\xi}_b^+.
    \end{equation}
\end{definition}

\begin{definition}[Equivalent Covariance Matrices]
    \label{def:covariance_equ}
    The covariance matrices $\mathbf{P}_a$ and $\mathbf{P}_b$ corresponding to the $\boldsymbol{\xi}_a$ and $\boldsymbol{\xi}_b$ are said to be equivalent if
    \begin{equation}
        \label{eq:equ_switch}
        \mathbf{P}_a = \mathbf{A}(\hat{\mathbf{x}}) \mathbf{P}_b \mathbf{A}^T(\hat{\mathbf{x}}),
    \end{equation}
    which is referred to as the equivalent covariance switch.

    Specifically, if the current system state is the predicted state $\hat{\mathbf{x}}^-$, then the equivalent covariances $\mathbf{P}_a^-$ and $\mathbf{P}_b^-$ satisfy
    \begin{equation}
        \mathbf{P}_a^- = \mathbf{A}(\hat{\mathbf{x}}^-) \mathbf{P}_b^- \mathbf{A}^T(\hat{\mathbf{x}}^-).
    \end{equation}
    Similarly, if the current system state is the updated state $\hat{\mathbf{x}}^+$, the equivalent covariances $\mathbf{P}_a^+$ and $\mathbf{P}_b^+$ satisfy
    \begin{equation}
        \mathbf{P}_a^+ = \mathbf{A}(\hat{\mathbf{x}}^+) \mathbf{P}_b^+ \mathbf{A}^T(\hat{\mathbf{x}}^+).
    \end{equation}
\end{definition}
\section{Covariance Switch and its Effect in Kalman Filters}
When the observation model mismatches with the process model in terms of error definitions, the covariance switch can be introduced~\cite{Han2024,Han2025master}.

Figure~\ref{fig:flowchat_sc_a_b} presents a demo of covariance switch in Kalman filters, in which ESKF(b) serves as the original filter~\cite{han2024covariance}. As shown in Fig.~\ref{fig:flowchat_sc_a_b}, $\mathbf{P}_{b,s}^-$ can be converted into the objective $\mathbf{P}_{a,s}^-$ using the forward covariance switch \eqref{eq:sc_forward}, and the subscript $s$ denotes the application of covariance switch. Then, the covariance and state update operation of ESKF(a) are performed. After updating the covariance matrix $\mathbf{P}_{a,s}$, it is crucial to transform it back to the corresponding matrix $\mathbf{P}_{b,s}$ using the updated states in \eqref{eq:sc_backward}. This backward covariance switch is necessary for the next covariance propagation.
\begin{equation}
    \label{eq:sc_forward}
    \mathbf{P}_{a,s}^- = \mathbf{A}(\hat{\mathbf{x}}^-) \mathbf{P}_{b,s}^- \mathbf{A}^{T}(\hat{\mathbf{x}}^-),
\end{equation}
\begin{equation}
    \label{eq:sc_backward}
    \mathbf{P}_{b,s}^+ = \mathbf{A}^{-1}(\hat{\mathbf{x}}^+) \mathbf{P}_{a,s}^+ \mathbf{A}^{-T}(\hat{\mathbf{x}}^+).
\end{equation}
According to Definition~\ref{def:covariance_equ}, both \eqref{eq:sc_forward} and \eqref{eq:sc_backward} represent equivalent covariance switches.
\begin{figure}[htb!]
    \centering
    \includegraphics[width=0.45\textwidth]{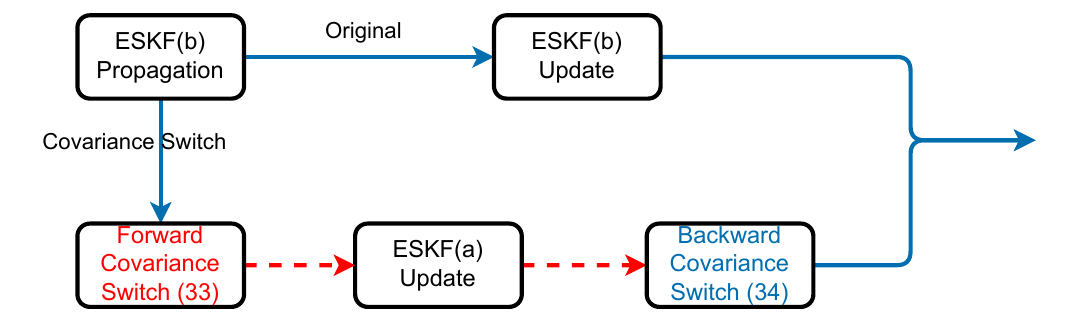}
    \caption{An example of covariance switch.}
    \label{fig:flowchat_sc_a_b}
\end{figure}

The general procedure for the covariance switch in a ESKF filtering algorithm is detailed in Algorithm~\ref{algo:covariance_swicht_a_b}, where $\mathbf{K}_a$ denotes the Kalman gain corresponding to the error definition $\boldsymbol{\xi}_a$, $\mathbf{R}$ is the observation noise covariance matrix independent of the error states, and $t_{k-1}$ denotes the previous update instant. Hereafter, ESKF(b)-$\mathrm{S}_{b\rightarrow a}$ refers to the covariance switch from $b$ to $a$, and similarly, ESKF(a)-$\mathrm{S}_{a\rightarrow b}$ refers to the switch from $a$ to $b$.

\begin{theorem}[{
Equivalence of Error State and Covariance Kept through State Propagation
}]
    \label{thm:prop_equivalence}
    Suppose at time $t_s$, the following equalities hold
    \begin{equation}
        \label{eq:xi_init}
        \boldsymbol{\xi}_a(t_s) = \mathbf{A}(\hat{\mathbf{x}}(t_s)) \boldsymbol{\xi}_b(t_s),
    \end{equation}
    \begin{equation}
        \label{eq:switch_init}
        \mathbf{P}_a(t_s) = \mathbf{A}(\hat{\mathbf{x}}(t_s)) \mathbf{P}_b(t_s) \mathbf{A}^T(\hat{\mathbf{x}}(t_s)),
    \end{equation}
    and the state propagation satisfies
    \begin{equation}
        \label{eq:prop_equivalence}
        \hat{\mathbf{x}}(t_f) = \hat{\mathbf{x}}(t_s) + \int_{t_s}^{t_f} \mathbf{f}(\hat{\mathbf{x}}(t'), \mathbf{u}(t'), t') \, dt'.
    \end{equation}
    Then, for $t \in [t_s, t_f]$, the following equivalences hold
    \begin{equation*}
        \boldsymbol{\xi}_a(t) = \mathbf{A}(\hat{\mathbf{x}}(t)) \boldsymbol{\xi}_b(t), \mathbf{P}_a(t) = \mathbf{A}(\hat{\mathbf{x}}(t)) \mathbf{P}_b(t) \mathbf{A}^T(\hat{\mathbf{x}}(t)).
    \end{equation*}
    That is, the error states and covariance matrices of ESKF(a) and ESKF(b) remain equivalent throughout the propagation interval.
\end{theorem}
\begin{proof}
    The full proof is provided in Appendix~\ref{appx:proof_of_lemma2}.
\end{proof}

Note that Eqs.~\eqref{eq:xi_init}-\eqref{eq:prop_equivalence} do not stipulate whether the current estimated state is the predicted state or the updated state. Although the propagation of the error state is also discussed here, in practical applications based on error-state Kalman filtering, the error state is usually reset to zero after the state update.

Moreover, there is no restriction on the relationship between $t_s$ and $t_f$ during the proof, so the conclusions also apply to the forward propagation and backward propagation of the covariance matrix over time.

It is worth noting that the proof of the equivalence between the error state and covariance matrices in Theorem~\ref{thm:prop_equivalence} is established under the assumption of continuous-time propagation. However, in most practical applications, such as inertial navigation systems (INS), both the error state and covariance are typically propagated in discrete time. To bridge this gap, we now provide a proof of the equivalence of the error state and covariance matrices under discrete-time propagation in appendix~\ref{appx:proof_of_theorem6}.

Note that covariance switch would be unworkable as given by the following Lemma~\ref{thm:ineffectiveness_covariance_switch}.

\begin{algorithm}[hb!]
    \caption{Covariance Switch Embedded in Kalman Filters}
    \label{algo:covariance_swicht_a_b}
    \begin{algorithmic}[1]
        \Require $\hat{\mathbf{x}}_{b,s}^+(t_{k-1})$, $\boldsymbol{\xi}_{b,s}^+(t_{k-1})$, $\mathbf{P}_{b,s}^+(t_{k-1})$
        \Ensure $\hat{\mathbf{x}}_{b,s}^+(t_{k})$, $\mathbf{P}_{b,s}^+(t_{k})$
        \State \textbf{EKSF(b) Propagation:} 
            \[
            \begin{aligned}
                \hat{\mathbf{x}}_{b,s}^+(t_{k-1}) &\rightarrow \hat{\mathbf{x}}_{b,s}^-(t_{k}), \\
                \boldsymbol{\xi}_{b,s}^+(t_{k-1}) &\rightarrow \boldsymbol{\xi}_{b,s}^-(t_{k}), \\
                \mathbf{P}_{b,s}^+(t_{k-1}) &\rightarrow \mathbf{P}_{b,s}^-(t_{k})
            \end{aligned}
            \]
        \State \textbf{Forward Covariance Switch~\eqref{eq:sc_forward}:} 
            \[
            \mathbf{P}_{a,s}^-(t_{k}) = \mathbf{A}(\hat{\mathbf{x}}_{b,s}^-(t_{k})) \mathbf{P}_{b,s}^-(t_{k}) \mathbf{A}^T(\hat{\mathbf{x}}_{b,s}^-(t_{k}))
            \]
        \State \textbf{$\mathbf{P}_{a,s}$ Update:} Obtain $\mathbf{H}_a$ and $\mathbf{R}$ from observation; compute Kalman gain $\mathbf{K}_a$, then
            \[
            \mathbf{P}_{a,s}^+(t_{k}) = \mathbf{P}_{a,s}^-(t_{k}) - \mathbf{K}_a \mathbf{H}_a \mathbf{P}_{a,s}^-(t_{k})
            \]
        \State \textbf{$\boldsymbol{\xi}_{a,s}$ Update:} 
            \[
            \boldsymbol{\xi}_{a,s}^+(t_{k}) = \boldsymbol{\xi}_a^-(t_{k}) + \mathbf{K}_a \delta \mathbf{z}
            \]
            $\boldsymbol{\xi}_a^-(t_{k})$ is typically set as zero in indirect Kalman filter. 
        \State \textbf{State Update:} 
            \[ \hat{\mathbf{x}}_{b,s}^-(t_{k}) \boxplus \boldsymbol{\xi}_{a,s}^+(t_{k}) = \hat{\mathbf{x}}_{a,s}^+(t_{k}) \] in which $\boxplus$ denote the step of absorbing error state in EKSF(a), 
            \[ \hat{\mathbf{x}}_{b,s}^+(t_{k}) = \hat{\mathbf{x}}_{a,s}^+(t_{k}) \]
        \State \textbf{Backward Covariance Switch~\eqref{eq:sc_backward}:} 
            \[
            \begin{aligned}
                \mathbf{P}_{b,s}^+(t_{k}) 
                = \mathbf{A}^{-1}(\hat{\mathbf{x}}_{b,s}^+(t_{k})) \mathbf{P}_{a,s}^+(t_{k}) \mathbf{A}^{-T}(\hat{\mathbf{x}}_{b,s}^+(t_{k}))
            \end{aligned}
            \]
        \State \Return $\hat{\mathbf{x}}_{b,s}^+(t_{k}), \quad \mathbf{P}_{b,s}^+(t_{k})$
    \end{algorithmic}
\end{algorithm}
\begin{lemma}[Ineffectiveness of Covariance Switch~\cite{Han2024}]
    \label{thm:ineffectiveness_covariance_switch}
    If the predicted error states and covariance matrices of ESKF(a) and ESKF(b) are equivalent, namely,
    \begin{equation}
        \label{eq:xi_relation_pred}
        \boldsymbol{\xi}_a^-(t_k) = \mathbf{A}(\hat{\mathbf{x}}^-(t_k)) \boldsymbol{\xi}_b^-(t_k),
    \end{equation}
    \begin{equation}
        \label{eq:P_relation_pred}
        \mathbf{P}_a^-(t_k) = \mathbf{A}(\hat{\mathbf{x}}^-(t_k)) \mathbf{P}_b^-(t_k) \mathbf{A}^T(\hat{\mathbf{x}}^-(t_k)).
    \end{equation}
    Then, after a single-step update, the updated error states and covariance matrices satisfy
    \begin{equation}
        \label{eq:xi_relation_update}
        \boldsymbol{\xi}_a^+(t_k) = \mathbf{A}(\hat{\mathbf{x}}^-(t_k)) \boldsymbol{\xi}_b^+(t_k),
    \end{equation}
    \begin{equation}
        \label{eq:p_update_relation}
        \mathbf{P}_a^+(t_k) = \mathbf{A}(\hat{\mathbf{x}}^-(t_k)) \mathbf{P}_b^+(t_k) \mathbf{A}^T(\hat{\mathbf{x}}^-(t_k)),
    \end{equation}
    which indicates that the updated error states and covariance matrices are not equivalent according to Definitions~\ref{def:error_equ} and~\ref{def:covariance_equ}.
    
    \begin{proof}
    For simplicity, the time index $t_k$ is omitted. Based on~\eqref{eq:H_relation}, the observation matrices satisfy
    \begin{equation}
        \label{eq:H_relation_pred}
        \mathbf{H}_b = \mathbf{H}_a \mathbf{A}(\hat{\mathbf{x}}^-).
    \end{equation}

    The Kalman gains corresponding to the error states $\boldsymbol{\xi}_a$ and $\boldsymbol{\xi}_b$ can be computed as
    \begin{equation}
        \mathbf{K}_a = \mathbf{P}_a^- \mathbf{H}_a^T \left( \mathbf{H}_a \mathbf{P}_a^- \mathbf{H}_a^T + \mathbf{R} \right)^{-1},
    \end{equation}
    \begin{equation}
        \mathbf{K}_b = \mathbf{P}_b^- \mathbf{H}_b^T \left( \mathbf{H}_b \mathbf{P}_b^- \mathbf{H}_b^T + \mathbf{R} \right)^{-1}.
    \end{equation}
    Based on~\eqref{eq:P_relation_pred} and~\eqref{eq:H_relation_pred}, and following the derivation in~\cite{Han2024}, the relationship between $\mathbf{K}_a$ and $\mathbf{K}_b$ can be obtained as
    \begin{equation}
        \label{eq:K_relation_pred}
        \mathbf{K}_a = \mathbf{A}(\hat{\mathbf{x}}^-) \mathbf{K}_b.
    \end{equation}
    Therefore, the updated error states $\boldsymbol{\xi}_a^+$ and $\boldsymbol{\xi}_b^+$ satisfy
    \begin{equation}
        \begin{aligned}
            \boldsymbol{\xi}_a^+ &= \boldsymbol{\xi}_a^- + \mathbf{K}_a \delta \mathbf{z} \\
            &= \mathbf{A}(\hat{\mathbf{x}}^-) \boldsymbol{\xi}_b^- + \mathbf{A}(\hat{\mathbf{x}}^-) \mathbf{K}_b \delta \mathbf{z} \\
            &= \mathbf{A}(\hat{\mathbf{x}}^-)(\boldsymbol{\xi}_b^- + \mathbf{K}_b \delta \mathbf{z}) \\
            &= \mathbf{A}(\hat{\mathbf{x}}^-) \boldsymbol{\xi}_b^+.
        \end{aligned}
    \end{equation}
    From the covariance update step in the Kalman filter~\cite{crassidis2004optimal}, combined with~\eqref{eq:P_relation_pred}, \eqref{eq:H_relation_pred}, and \eqref{eq:K_relation_pred}, the relationship between $\mathbf{P}_a^+$ and $\mathbf{P}_b^+$ is derived as
    \begin{equation}
        \label{eq:P_relation_update}
        \begin{aligned}
            \mathbf{P}_a^+ & = \mathbf{P}_a^- - \mathbf{K}_a \mathbf{H}_a \mathbf{P}_a^- \\
            & = \mathbf{A}(\hat{\mathbf{x}}^-) \mathbf{P}_b^- \mathbf{A}^T(\hat{\mathbf{x}}^-) - \mathbf{A}(\hat{\mathbf{x}}^-) \mathbf{K}_b \mathbf{H}_b \mathbf{P}_b^- \mathbf{A}^T(\hat{\mathbf{x}}^-) \\
            & = \mathbf{A}(\hat{\mathbf{x}}^-) \left( \mathbf{P}_b^- - \mathbf{K}_b \mathbf{H}_b \mathbf{P}_b^- \right) \mathbf{A}^T(\hat{\mathbf{x}}^-) \\
            & = \mathbf{A}(\hat{\mathbf{x}}^-) \mathbf{P}_b^+ \mathbf{A}^T(\hat{\mathbf{x}}^-),
        \end{aligned}
    \end{equation}
which indicates that although the covariances $\mathbf{P}_a^-$ and $\mathbf{P}_b^-$ of ESKF(a) and ESKF(b) before the update are equivalent, the updated covariances $\mathbf{P}_a^+$ and $\mathbf{P}_b^+$ do not satisfy the equivalence condition. Based on forward covariance switch~\eqref{eq:sc_forward}, $\mathbf{P}_{a,s}^-$ and $\mathbf{P}_{b}^-$ satisfy the condition given in~\eqref{eq:P_relation_pred}. Consequently, $\mathbf{P}_{a,s}^+$ and $\mathbf{P}_{b}^+$ satisfy a same relation as in~\eqref{eq:P_relation_update}, i.e.,
\begin{equation}
    \label{eq:relation_a_as}
    \mathbf{P}_{a,s}^+ = \mathbf{P}_{a}^+ = \mathbf{A}(\hat{\mathbf{x}}^-) \mathbf{P}_b^+ \mathbf{A}^T(\hat{\mathbf{x}}^-).
\end{equation}
In contrast to step $6$ in algorithm~\ref{algo:covariance_swicht_a_b}, if the $\mathbf{A}(\hat{\mathbf{x}}^-(t_k))$ is used for backward covariance switch~\eqref{eq:sc_backward}, then 
\begin{equation}
\begin{aligned}
    \mathbf{P}_{b,s}^+ &= \mathbf{A}^{-1}(\hat{\mathbf{x}}^-) \mathbf{P}_{a,s}^+ \mathbf{A}^{-T}(\hat{\mathbf{x}}^-) \\
     &= \mathbf{A}^{-1}(\hat{\mathbf{x}}^-) \mathbf{P}_{a}^+ \mathbf{A}^{-T}(\hat{\mathbf{x}}^-) \\
    &= \mathbf{P}_b^+.
\end{aligned}
\end{equation}
In this case, the introduction of covariance switch is unworkable.
\end{proof}
\end{lemma}

From Theorem~\ref{thm:prop_equivalence}, it follows that under the condition of equivalent initial error states and covariance matrices, the propagation results of different ESKF algorithms preserve the equivalence of error states and covariances. Consequently, the assumptions of Lemma~\ref{thm:ineffectiveness_covariance_switch} are satisfied. Furthermore, Lemma~\ref{thm:ineffectiveness_covariance_switch} implies that the updated error states of different ESKF algorithms satisfy the relationship given by equation~\eqref{eq:xi_relation_update}.

Reference~\cite{han2024covariance} presented the equivalence between the left- and right-invariant system state updates. Here, we provide a more general form applicable to all ESKF variants.
\begin{lemma}[Relationship of State Update Under Different Error State Definitions]
    \label{thm:x_correct_equivalence}
    If two error states $\boldsymbol{\xi}_a$ and $\boldsymbol{\xi}_b$ satisfy the relationship in~\eqref{eq:xi_relation_update}, then the updated system states are identical, i.e.,
    \begin{equation}
        \label{eq:a_b_update_relation}
        \hat{\mathbf{x}}_a^+ = \hat{\mathbf{x}}_b^+.
    \end{equation}
    \begin{proof}
    Consider the error state $\boldsymbol{\xi}_a$ and the additive error state $\delta \mathbf{x}$ which satisfy the following transformation
    \begin{equation}
        \boldsymbol{\xi}_a = \mathbf{A}_{ekf \rightarrow a}(\hat{\mathbf{x}}) \delta \mathbf{x}.
    \end{equation}
    In EKF, the error state is defined as $\delta \mathbf{x} = \hat{\mathbf{x}} - \mathbf{x}$. Thus, it follows that
    \begin{equation}
        \begin{aligned}
            \boldsymbol{\xi}_a = \mathbf{A}_{ekf \rightarrow a}(\hat{\mathbf{x}}) \delta \mathbf{x} = \mathbf{A}_{ekf \rightarrow a}(\hat{\mathbf{x}}) (\hat{\mathbf{x}} - \mathbf{x}).
        \end{aligned}
    \end{equation}
    Consequently, the state update for $\boldsymbol{\xi}_a$ can be expressed as
    \begin{equation}
        \label{eq:xi_update}
        \mathbf{x} = \hat{\mathbf{x}} - \mathbf{A}_{ekf \rightarrow a}^{-1}(\hat{\mathbf{x}}) \boldsymbol{\xi}_a.
    \end{equation}
    According to Lemma~\ref{thm:ineffectiveness_covariance_switch}, after one-step update, $\boldsymbol{\xi}_a^+$ and $\delta \mathbf{x}^+$ satisfy
        \begin{equation}
            \label{eq:xi_x_relation}
            \boldsymbol{\xi}_a^+ = \mathbf{A}_{ekf \rightarrow a}(\hat{\mathbf{x}}^-) \delta \mathbf{x}^+.
        \end{equation}
    The relationship between the updated system states $\hat{\mathbf{x}}^+$ and $\hat{\mathbf{x}}_a^+$, updated respectively by the error states $\delta \mathbf{x}^+$ and $\boldsymbol{\xi}_a^+$ based on \eqref{eq:xi_update}, is given by
    \begin{equation}
        \begin{aligned}
            \hat{\mathbf{x}}_a^+ &= \hat{\mathbf{x}}^- - \mathbf{A}_{ekf \rightarrow a}^{-1}(\hat{\mathbf{x}}^-) \boldsymbol{\xi}_a^+ \\
            &= \hat{\mathbf{x}}^- - \delta \mathbf{x}^+ \\
            &= \hat{\mathbf{x}}_{ekf}^+.
        \end{aligned}
    \end{equation}
    Similarly, it can be shown that
    \begin{equation}
        \hat{\mathbf{x}}_b^+ = \hat{\mathbf{x}}_{ekf}^+ = \hat{\mathbf{x}}_a^+.
    \end{equation}
\end{proof}
\end{lemma}

\begin{theorem}\label{pr:result_equ}
    Under given initial conditions, i.e., identical system states and equivalent covariance matrices, the estimated system states $\hat{\mathbf{x}}_a^+$ and $\hat{\mathbf{x}}_b^+$ after the first update step of different ESKF algorithms are identical, regardless of the propagation duration. 
    However, $\mathbf{P}_a^+$ and $\mathbf{P}_b^+$ represent the updated system uncertainties, which are inherently mismatched because of the predicted state in $\mathbf{A}(\hat{\mathbf{x}})$.
    
    \begin{proof}
        According to Theorem~\ref{thm:prop_equivalence}, given that the system states are identical and covariance matrices are equivalent at the initial time, their corresponding covariance matrices and error states remain equivalent through arbitrary propagation time.

        By Lemma~\ref{thm:ineffectiveness_covariance_switch}, when the error states and covariance matrices of different ESKF algorithms are equivalent prior to update, the single-step updated error states and covariance matrices satisfy relations \eqref{eq:xi_relation_update} and \eqref{eq:p_update_relation}.

        Furthermore, from Lemma~\ref{thm:x_correct_equivalence}, their updated system states coincide.
    \end{proof}
\end{theorem}

\begin{lemma}[Effectiveness of Covariance Switch~\cite{Han2024}]
    \label{thm:effectiveness_covariance_switch}
    Under the condition of identical system states and equivalent covariance matrices, by introducing the forward covariance switch \eqref{eq:sc_forward} and the backward covariance switch \eqref{eq:sc_backward}, the estimated system states will be identical and covariance will be equivalent between ESKF(b)-$\mathrm{S}_{b\rightarrow a}$ and ESKF(a).
    
    \begin{proof}
    First, consider the single-step propagation followed by a single-step update. Under given initial conditions, according to Theorem~\ref{thm:prop_equivalence}, right before the update at $t_k$, we have
    \begin{equation}
    \label{eq:equ}
    \begin{aligned}
        \mathbf{x}_a^-(t_k) &= \mathbf{x}_{b,s}^-(t_k),  \\
        \boldsymbol{\xi}_a^-(t_k) &= \mathbf{A}(\hat{\mathbf{x}}^-(t_k)) \boldsymbol{\xi}_{b,s}^-(t_k), \\
        \mathbf{P}_a^-(t_k) &= \mathbf{A}(\hat{\mathbf{x}}^-(t_k)) \mathbf{P}_{b,s}^-(t_k) \mathbf{A}^T(\hat{\mathbf{x}}^-(t_k)).
    \end{aligned}
    \end{equation}
    From the definition of the forward covariance switch \eqref{eq:sc_forward}, it follows that
    \begin{equation*}
        \mathbf{P}_a^-(t_k) = \mathbf{P}_{a,s}^-(t_k).
    \end{equation*}
    Therefore, based on Lemmas~\ref{thm:ineffectiveness_covariance_switch} and~\ref{thm:x_correct_equivalence}, it is straightforward to obtain
    \begin{equation*}
    \begin{aligned}
        \mathbf{x}_a^+(t_k) &= \mathbf{x}_{b,s}^+(t_k).
    \end{aligned}
    \end{equation*}
    Based on~\eqref{eq:relation_a_as} and backward covariance switch \eqref{eq:sc_backward} in step $6$ of Algorithm~\ref{algo:covariance_swicht_a_b},
    \begin{equation}
        \label{eq:relation_a_bs}
        \begin{aligned}
            \mathbf{P}_{b,s}^+ &= \mathbf{A}^{-1}(\hat{\mathbf{x}}^+(t_k)) \mathbf{P}_{a,s}^+ \mathbf{A}^{-T}(\hat{\mathbf{x}}^+(t_k)) \\
            &= \mathbf{A}^{-1}(\hat{\mathbf{x}}^+(t_k)) \mathbf{P}_{a}^+ \mathbf{A}^{-T}(\hat{\mathbf{x}}^+(t_k)).
        \end{aligned}   
    \end{equation}
    Based on Theorem~\ref{thm:prop_equivalence}, after a single-step propagation, the relation in \eqref{eq:equ} still holds at time $t_{k+1}$. Therefore, by extending the analysis to multiple steps of propagation and update, it follows that ESKF(b)-$\mathrm{S}_{b\rightarrow a}$ yields identical state estimation to ESKF(a).

    If the backward covariance switch uses the predicted system state for the transformation, then after the update at time $t_k$, ESKF(a) and ESKF(b)-$\mathrm{S}_{b\rightarrow a}$ satisfy the following relations
    \begin{equation}
        \begin{aligned}
            \mathbf{x}_a^+(t_k) &= \mathbf{x}_{b,s}^+(t_k), \\
            \mathbf{P}_a^+(t_k) &= \mathbf{A}(\hat{\mathbf{x}}^-(t_k)) \mathbf{P}_{b,s}^+(t_k) \mathbf{A}^T(\hat{\mathbf{x}}^-(t_k)).
        \end{aligned}
    \end{equation}
    Based on the derivations in Theorem~\ref{thm:prop_equivalence}, the predicted system states at $t_{k+1}$ remain identical, but the covariance matrices no longer satisfy the equivalence relation as follows
    \begin{equation}
        \hat{\mathbf{x}}_a^-(t_{k+1}) = \hat{\mathbf{x}}_{b}^-(t_{k+1}),
    \end{equation}
    \begin{equation}
        \mathbf{P}_a^-(t_{k+1}) \neq \mathbf{A}(\hat{\mathbf{x}}^-(t_{k+1})) \mathbf{P}_{b,s}^-(t_{k+1}) \mathbf{A}^T(\hat{\mathbf{x}}^-(t_{k+1})).
    \end{equation}
    \end{proof}
\end{lemma}
According to Lemma~\ref{thm:effectiveness_covariance_switch}, the backward covariance switch must be performed using the updated system state as specified in \eqref{eq:sc_backward}.
\section{Covariance Transformation}
Based on the Lemma~\ref{thm:effectiveness_covariance_switch}, it is established that ESKF(b)-$\mathrm{S}_{b\rightarrow a}$ and ESKF(a) yield identical estimated system states. We now consider the difference between ESKF(b)-$\mathrm{S}_{b\rightarrow a}$ and ESKF(b). Under given initial conditions, after the first update, theorem~\ref{pr:result_equ} implies that $\hat{\mathbf{x}}_b^+ = \hat{\mathbf{x}}_a^+ = \hat{\mathbf{x}}_{b,s}^+$; however, the updated covariance matrices $\mathbf{P}_{b,s}^+$ and $\mathbf{P}_{b}^+$ differ. From equations \eqref{eq:p_update_relation} and~\eqref{eq:relation_a_bs}, it follows that
\begin{equation}
    \begin{aligned}
        &\mathbf{P}_{b,s}^+ \\
        &= \mathbf{A}^{-1}(\hat{\mathbf{x}}^+) \mathbf{P}_a^+ \mathbf{A}^{-T}(\hat{\mathbf{x}}^+) \\
        &= \mathbf{A}^{-1}(\hat{\mathbf{x}}^+) (\mathbf{A}(\hat{\mathbf{x}}^-) \mathbf{P}_b^+ \mathbf{A}^{T}(\hat{\mathbf{x}}^-)) \mathbf{A}^{-T}(\hat{\mathbf{x}}^+) \\
        &= T_{b \to a}(\hat{\mathbf{x}}^+, \hat{\mathbf{x}}^-) \mathbf{P}_b^+ T_{b \to a}^T(\hat{\mathbf{x}}^+, \hat{\mathbf{x}}^-),
    \end{aligned}
\end{equation}
\begin{equation}
    T_{b \to a}(\hat{\mathbf{x}}^+, \hat{\mathbf{x}}^-) = \mathbf{A}^{-1}(\hat{\mathbf{x}}^+) \mathbf{A}(\hat{\mathbf{x}}^-).
\end{equation}

Therefore, the only difference between ESKF(b)-$\mathrm{S}_{b\rightarrow a}$ and ESKF(b) lies in the covariance matrix after each update~\cite{Han2024,Han2025master}. As stated in~\cite{Han2024}, the covariance switch does not alter the updated state; it only modifies the updated covariance matrix. By applying a single covariance transformation $T_{b \to a}$ to the covariance matrix of ESKF(b) after each update, one can obtain system state estimates identical to those of ESKF(b)-$\mathrm{S}_{b\rightarrow a}$ and ESKF(a). This single-step transformation is referred to as the covariance transformation. We next use CT-ESKF to denote the ESKF algorithm incorporating the covariance transformation and let $\mathbf{P}_{b,CT}$ denote the covariance matrix of $\mathrm{CT}_{b \rightarrow a}$-ESKF(b), 
\begin{equation}
    \label{eq:transform_b_a}
        \mathbf{P}_{b,CT}^+ = \mathbf{P}_{b,s}^+ = T_{b \to a}(\hat{\mathbf{x}}^+, \hat{\mathbf{x}}^-) \mathbf{P}_b^+ T_{b \to a}^T(\hat{\mathbf{x}}^+, \hat{\mathbf{x}}^-).
\end{equation}
where the subscript $CT$ indicates the covariance transformation. The detailed algorithmic procedure is given in Algorithm~\ref{algo:covariance_transform_a_b}. Without causing ambiguity, $\mathrm{CT}_{b \rightarrow a}$-ESKF(b) refers to the covariance transformation from $b$ to $a$, and similarly, $\mathrm{CT}_{a \rightarrow b}$-ESKF(a) refers to the covariance transformation from $a$ to $b$.

We note that a similar idea also appeared in~\cite{song2024affine}, but its is restricted to the update steps of the specific VIO problem. The CT-ESKF is more general and can be applied to any state estimation problems.

\begin{algorithm}[hb!]
    \caption{Covariance Transformation Procedure} 
    \label{algo:covariance_transform_a_b}
    \begin{algorithmic}[1] 
        \Require $\hat{\mathbf{x}}_b^+(t_{k-1})$, $\boldsymbol{\xi}_{b}^+(t_{k-1})$, $\mathbf{P}_{b}^+(t_{k-1})$
        \Ensure updated $\hat{\mathbf{x}}_{b,CT}^+(t_k)$, $\mathbf{P}_{b,CT}^+(t_k)$
        \State \textbf{Propagation:}
        \[
        \begin{aligned}
            \hat{\mathbf{x}}_b^+(t_{k-1}) &\rightarrow  \hat{\mathbf{x}}_{b}^-(t_k),\\
            \boldsymbol{\xi}_{b}^+(t_{k-1}) &\rightarrow \boldsymbol{\xi}_{b}^-(t_k),\\ \mathbf{P}_{b}^+(t_{k-1}) &\rightarrow \mathbf{P}_{b}^-(t_k)
        \end{aligned}
        \]
        \State \textbf{$\mathbf{P}_b$ Update:} From the observation, obtain $\mathbf{H}_b$ and $\mathbf{R}$, then compute $\mathbf{K}_b$. Update as
        \[
        \mathbf{P}_{b}^+(t_k) = \mathbf{P}_{b}^-(t_k) - \mathbf{K}_b \mathbf{H}_b \mathbf{P}_{b}^-(t_k)
        \]
        \State \textbf{$\boldsymbol{\xi}_b$ Update:}
        \[
        \boldsymbol{\xi}_{b}^+(t_k) = \boldsymbol{\xi}_{b}^-(t_k) + \mathbf{K}_b \delta \mathbf{z}
        \]
        \State \textbf{State Update:}
        \[ \hat{\mathbf{x}}_{b}^-(t_{k}) \boxplus \boldsymbol{\xi}_{b}^+(t_{k}) = \hat{\mathbf{x}}_{b}^+(t_{k}) \] in which $\boxplus$ denote the step of absorbing error state in EKSF(b), 
        \State \textbf{Covariance Transformation~\eqref{eq:transform_b_a}:} 
        \[
        \begin{aligned}
            &\mathbf{P}_{b,CT}^+(t_k) \\
            &= T_{b \to a}(\hat{\mathbf{x}}^+(t_{k}), \hat{\mathbf{x}}^-(t_{k})) \mathbf{P}_b^+(t_{k}) T_{b \to a}^T(\hat{\mathbf{x}}^+(t_{k}), \hat{\mathbf{x}}^-(t_{k}))
        \end{aligned}
        \]
        \State \Return $\hat{\mathbf{x}}_{b,CT}^+(t_k)$, $\mathbf{P}_{b,CT}^+(t_k)$
    \end{algorithmic}
\end{algorithm}

\begin{theorem}[Covariance Transformation Relationships]\label{thm:P_relation_ct_s}
    The covariances of ESKF(a), ESKF(b), $\mathrm{CT}_{a \rightarrow b}$-ESKF(a), and $\mathrm{CT}_{b \rightarrow a}$-ESKF(b), denoted as $\mathbf{P}_a$, $\mathbf{P}_b$, $\mathbf{P}_{a,CT}$, and $\mathbf{P}_{b,CT}$, respectively, are related by invertible transformations. The covariance transformation matrices in CT-ESKF are given by
    \begin{equation}
        \label{eq:CT_first}
        T_{b \to a}(\hat{\mathbf{x}}^+, \hat{\mathbf{x}}^-) = \mathbf{A}^{-1}(\hat{\mathbf{x}}^+) \mathbf{A}(\hat{\mathbf{x}}^-),
    \end{equation}
    \begin{equation}
        \label{eq:CT_second}
        T_{a \to b}(\hat{\mathbf{x}}^+, \hat{\mathbf{x}}^-) = \mathbf{A}(\hat{\mathbf{x}}^+) \mathbf{A}^{-1}(\hat{\mathbf{x}}^-).
    \end{equation}
    \begin{proof}
    As illustrated in Fig.~\ref{fig:P_transform_a_b}, the bidirectional arrows indicate that the covariances can be transformed forth and back between each other.
    \begin{figure}[ht!]
        \centering
        \includegraphics[width=0.5\textwidth]{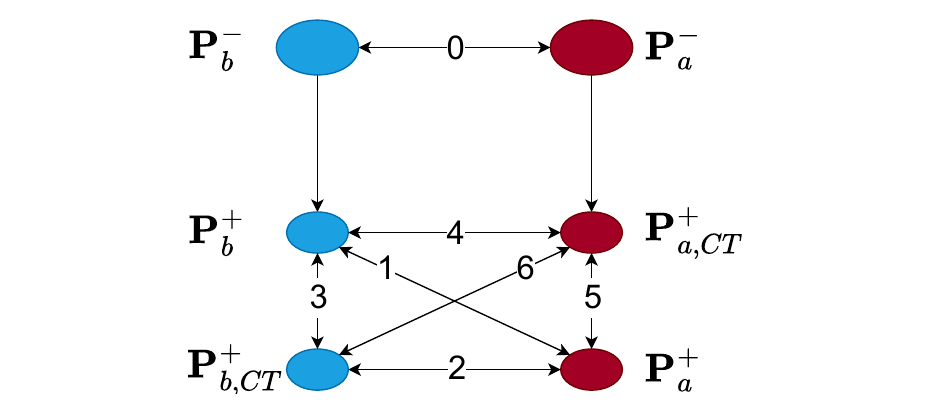}
        \caption{Transformation relationships (0-6) among different ESKF algorithms.}
        \label{fig:P_transform_a_b}
    \end{figure}
    The transformation relationship labeled as $0$ refers to the correspondence between the predicted covariances $\mathbf{P}_b^-$ and $\mathbf{P}_a^-$, as expressed in~\eqref{eq:P_relation_pred}.

Transformation relationship $1$ denotes the relationship between the updated covariances $\mathbf{P}_b^+$ and $\mathbf{P}_a^+$ after a single update step. According to Lemma~\ref{thm:ineffectiveness_covariance_switch}, this relationship satisfies~\eqref{eq:p_update_relation}.

Based on~\eqref{eq:relation_a_as}, $\mathbf{P}_{a}^+$ is equal to $\mathbf{P}_{a,s}^+$, and according to~\eqref{eq:transform_b_a}, $\mathbf{P}_{b,CT}^+$ equals $\mathbf{P}_{b,s}^+$. Therefore, transformation relationship~2 also describes the connection between $\mathbf{P}_{b,s}^+$ and $\mathbf{P}_{a,s}^+$, corresponding to the backward covariance switch in ESKF(b)-$\mathrm{S}_{b \to a}$, which satisfies~\eqref{eq:sc_backward}:
\begin{equation}
    \mathbf{P}_{b,CT}^+ = \mathbf{A}^{-1}(\hat{\mathbf{x}}^+) \mathbf{P}_{a}^+ \mathbf{A}^{-T}(\hat{\mathbf{x}}^+).
\end{equation}

Transformation relationship $3$ represents the covariance transformation from ESKF(b) to ESKF(a), expressed in~\eqref{eq:transform_b_a}. Based on the directional flow indicated by the arrows in the diagram, transformation $3$ can be derived by combining transformation relationships $1$ and $2$.

Similarly, transformation relationship $4$ corresponds to the backward covariance switch in ESKF(a)-$\mathrm{S}_{a\rightarrow b}$, which can be expressed as
\begin{equation}
    \mathbf{P}_{a,CT}^+ = \mathbf{P}_{a,s}^+ = \mathbf{A}(\hat{\mathbf{x}}^+) \mathbf{P}_b^+ \mathbf{A}^{T}(\hat{\mathbf{x}}^+).
\end{equation}

Transformation relationship $5$ represents the covariance transformation in $\mathrm{CT}_{a \rightarrow b}$-ESKF(a). By combining transformation relationships $1$ and $4$, we obtain
\begin{equation}
    \begin{aligned}
        \mathbf{P}_{a,CT}^+ &= \mathbf{A}(\hat{\mathbf{x}}^+) \mathbf{P}_b^+ \mathbf{A}^{T}(\hat{\mathbf{x}}^+) \\
        &= \mathbf{A}(\hat{\mathbf{x}}^+) \left(\mathbf{A}^{-1}(\hat{\mathbf{x}}^-) \, \mathbf{P}_a^+ \, \mathbf{A}^{-T}(\hat{\mathbf{x}}^-)\right) \mathbf{A}^{T}(\hat{\mathbf{x}}^+) \\
        &= T_{a \to b}(\hat{\mathbf{x}}^+,\hat{\mathbf{x}}^-) \mathbf{P}_a^+  T_{a \to b}^T(\hat{\mathbf{x}}^+,\hat{\mathbf{x}}^-).
    \end{aligned}
\end{equation}
Transformation relationship $6$ describes the connection between $\mathbf{P}_{a,CT}^+$ and $\mathbf{P}_{b,CT}^+$. Based on transformation relationships $2$ and $5$, it follows that
\begin{equation}
    \begin{aligned}
        &\mathbf{P}_{a,CT}^+ \\
        &= \mathbf{A}(\hat{\mathbf{x}}^+) \left( \mathbf{A}^{-1}(\hat{\mathbf{x}}^-) \, \mathbf{P}_a^+ \, \mathbf{A}^{-T}(\hat{\mathbf{x}}^-) \right) \mathbf{A}^{T}(\hat{\mathbf{x}}^+) \\
        &= \mathbf{A}(\hat{\mathbf{x}}^+) \mathbf{A}^{-1}(\hat{\mathbf{x}}^-) \, \mathbf{A}(\hat{\mathbf{x}}^+) \, \mathbf{P}_{b,CT}^+ \, \mathbf{A}^{T}(\hat{\mathbf{x}}^+) \, \mathbf{A}^{-T}(\hat{\mathbf{x}}^-) \mathbf{A}^{T}(\hat{\mathbf{x}}^+).
    \end{aligned}
\end{equation}
\end{proof}
\end{theorem}

\begin{remark}
    It is important to note that covariance switch explicitly changes the current error state to another one, whereas the covariance transformation only need to adjust the distribution of the current covariance matrix. As shown in Theorem~\ref{thm:P_relation_ct_s}, when the update is performed by the conventional linearization-based Kalman filter, the updated covariance matrices of different error states satisfy the transformation relationship~1. Under this condition, covariance switch (transformation relationship 2) and covariance transformation (transformation relationship 3) yield identical results. However, when the ESKF incorporates iterative observation updates, transformation relationship~1 no longer holds, and two approaches lead to different estimation results.
\end{remark}

\begin{remark}
    The transformation matrix $\mathbf{A}(\hat{\mathbf{x}}^-)$ between $\boldsymbol{\xi}_a$ and $\boldsymbol{\xi}_b$ is not required to have a specific form; it only needs to be invertible. Consequently, the aforementioned results are applicable to various types of ESKF implementations with different error state definitions, such as the standard EKF, InEKF, state transformation EKF (ST-EKF)~\cite{wang2018state}, and Equivariant Filter (EqF)~\cite{van2022equivariant}, without any restriction on the state dimension. Therefore, the proposed CT-ESKF serves as a general framework that can incorporate any error-state Kalman filter as the original estimator, yielding different estimation results by adjusting the transformation matrix $\mathbf{T}$.

    Furthermore, if $\mathbf{T}$ is orthogonal, the transformation preserves the trace of the covariance matrix. Similarly, if $\det(\mathbf{T}) = \pm 1$, the determinant of the covariance matrix remains unchanged after the transformation.
\end{remark}

This implies that one can designate an ESKF with an arbitrary error state definition as the original filter and obtain potentially improved estimation performance by applying a single-step covariance transformation to adjust the updated covariance matrix. The matrix $\mathbf{T}$ depends on both the predicted and updated system states. When the Kalman gain is small, $\mathbf{T}$ will be close to be the identity matrix, resulting in little impact on the covariance matrix. Conversely, when the Kalman gain is large, the effect will become significant.

\begin{remark}
    
    If the invertible transformation matrix $\mathbf{A}$ between error states is independent of the system state, $\mathbf{T}$ in \eqref{eq:CT_first} and \eqref{eq:CT_second} become the identity matrix. Then, the updated covariance matrices of ESKF-(a) and ESKF-(b) are fully equivalent in terms of information and differ only in their distributional representations.
\end{remark}

So far, we have established the relationships among different ESKF formulations. While the connections between the L-InEKF and R-InEKF have been discussed in~\cite{ge2025difference,han2024covariance}, the analyses in this work are applicable to arbitrary ESKF variants.

Theorem~\ref{thm:P_relation_ct_s} shows that the covariance transformation matrix depends solely on the mapping between different error state definitions. This implies that the updated covariance reflects the system's uncertainty more accurately when the error state is appropriately designed.

Therefore, the error state formulation should be carefully constructed to better capture the geometric characteristics of the estimation problem, thereby enabling the updated covariance to more faithfully represent the system's actual uncertainty. In particular, by properly designing the transformation matrix $\mathbf{A}(\hat{\mathbf{x}})$ based on \eqref{eq:F_relation}, \eqref{eq:G_relation}, and \eqref{eq:H_relation}, one can ensure that the associated system matrices $\mathbf{F}_a$, $\mathbf{G}_a$, and $\mathbf{H}_a$ are independent of the trajectory. This leads to improved robustness and performance of the estimator.
\section{Preliminaries of EKF and InEKF}\label{sec:error_group}
The previous sections have addressed the general filtering methods with covariance switch and transformation. In this section, the theoretical analyses are validated by considering the EKF and InEKF as specific examples in integrated navigation.

\subsection{Linearization of the EKF Model}
This paper adopts the Earth-Centered Earth-Fixed (ECEF) frame, denoted as the $e$ frame, as the reference coordinate system. The navigation state to be estimated is given by \cite{groves2015principles}
\begin{equation}
    \mathbf{x}_{all} : \mathbf{C}_b^e, \quad \mathbf{v}_{eb}^e, \quad \mathbf{r}_{eb}^e, \quad \mathbf{b}_g, \quad \mathbf{b}_a,
\end{equation}
where $b$ denotes the body frame, $\mathbf{v}_{eb}^e$ and $\mathbf{r}_{eb}^e$ represent the velocity and position of the vehicle with respect to the $e$ frame, expressed in the $e$ frame. The matrix $\mathbf{C}_{\alpha}^{\beta}$ denotes the rotation from $\alpha$ frame to $\beta$ frame. $\mathbf{b}_g$ and $\mathbf{b}_a$ represent gyroscope and accelerometer biases, respectively. 

The IMU measurement errors are modeled as
\begin{align}
    \tilde{\boldsymbol{\omega}}_{ib}^b - \boldsymbol{\omega}_{ib}^b &= \mathbf{b}_g + \mathbf{w}_g, \\
    \tilde{\mathbf{f}}_{ib}^b - \mathbf{f}_{ib}^b &= \mathbf{b}_a + \mathbf{w}_a,
\end{align}
where $\tilde{\boldsymbol{\omega}}_{ib}^b$ and $\tilde{\mathbf{f}}_{ib}^b$ are the measured angular velocity and specific force by the IMU, respectively. The subscript $i$ denotes the inertial frame (Earth-Centered Inertial frame). The noise terms $\mathbf{w}_g$ and $\mathbf{w}_a$ correspond to the measurement noise of the gyroscope and accelerometer, respectively.

The INS kinematic model expressed in the $e$ frame is given by \cite{groves2015principles}
\begin{equation}
\begin{aligned}
    \dot{\mathbf{C}}_b^e &= \mathbf{C}_b^e (\boldsymbol{\omega}_{ib}^b \times) - (\boldsymbol{\omega}_{ie}^e \times) \mathbf{C}_b^e, \\
    \dot{\mathbf{v}}_{eb}^e &= \mathbf{C}_b^e \mathbf{f}_{ib}^b - 2(\boldsymbol{\omega}_{ie}^e \times)\mathbf{v}_{eb}^e + \mathbf{g}^e, \\
    \dot{\mathbf{r}}_{eb}^e &= \mathbf{v}_{eb}^e,  \\
    \dot{\mathbf{b}}_g &= \mathbf{w}_{bg}, \\
    \dot{\mathbf{b}}_a &= \mathbf{w}_{ba},
\end{aligned}
\end{equation}
where $(\cdot \times)$ denotes the skew-symmetric matrix associated with a vector. $\mathbf{w}_{bg}$ and $\mathbf{w}_{ba}$ represent the noise of the gyroscope bias and accelerometer bias, respectively. $\mathbf{g}^e$ denotes the gravity vector in the $e$ frame, and $\boldsymbol{\omega}_{ie}^e$ is the earth rotation rate expressed in the $e$ frame.

The EKF error state $\delta \mathbf{x}$ includes errors in attitude, velocity, and position. The augmented error state including the gyroscope and accelerometer biases is denoted as $\delta \mathbf{x}_{all}$, given by
\begin{equation}
    \delta \mathbf{x} = \left[\mathbf{\varphi}^T \quad \delta \mathbf{v}^{eT} \quad \delta \mathbf{r}^{eT} \right]^T,
\end{equation}
\begin{equation}
    \delta \mathbf{x}_{all} = \left[\delta \mathbf{x}^T \quad \delta\mathbf{b}_g^T \quad \delta\mathbf{b}_a^T \right]^T,
\end{equation}
where the error states are specifically defined as
\begin{equation}
\label{eq:delta_x_definition_E}
\begin{aligned}
    \hat{\mathbf{C}}_b^e \mathbf{C}_e^b &\approx \mathbf{I}_3 + \mathbf{\varphi} \times, \\
    \delta \mathbf{v}^e &= \hat{\mathbf{v}}_{eb}^e - \mathbf{v}_{eb}^e, \\
    \delta \mathbf{r}^e &= \hat{\mathbf{r}}_{eb}^e - \mathbf{r}_{eb}^e, \\
    \delta \mathbf{b}_g &= \hat{\mathbf{b}}_g - \mathbf{b}_g, \\
    \delta \mathbf{b}_a &= \hat{\mathbf{b}}_a - \mathbf{b}_a.
\end{aligned}
\end{equation}
The linearized system model of the EKF is given by~\cite{groves2015principles}
\begin{equation}
    \delta \dot{\mathbf{x}}_{all} = \mathbf{F}_{ekf} \, \delta \mathbf{x}_{all} + \mathbf{G}_{ekf} \, \mathbf{w},
\end{equation}
where
\begin{equation}
\mathbf{F}_{ekf} = \begin{bmatrix}
-(\boldsymbol{\omega}_{ie}^e \times) & \mathbf{0}_{3} & \mathbf{0}_{3} & \hat{\mathbf{C}}_b^e & \mathbf{0}_{3} \\
-(\hat{\mathbf{C}}_b^e \hat{\mathbf{f}}_{ib}^b) \times & -2(\boldsymbol{\omega}_{ie}^e \times) & \mathbf{0}_{3} & \mathbf{0}_3 & \hat{\mathbf{C}}_b^e \\
\mathbf{0}_{3} & \mathbf{I}_{3} & \mathbf{0}_{3} & \mathbf{0}_{3} & \mathbf{0}_{3} \\
\mathbf{0}_{3} & \mathbf{0}_{3} & \mathbf{0}_{3} & \mathbf{0}_{3} & \mathbf{0}_{3} \\
\mathbf{0}_{3} & \mathbf{0}_{3} & \mathbf{0}_{3} & \mathbf{0}_{3} & \mathbf{0}_{3}
\end{bmatrix},
\end{equation}
\begin{equation}
\mathbf{G}_{ekf} =  \left[
\begin{array}{ccccc}
\hat{\mathbf{C}}_b^e & \mathbf{0}_{3} & \mathbf{0}_{3} & \mathbf{0}_{3} \\
\mathbf{0}_3 & \hat{\mathbf{C}}_b^e & \mathbf{0}_{3} & \mathbf{0}_{3} \\
\mathbf{0}_3 & \mathbf{0}_{3} & \mathbf{0}_{3} & \mathbf{0}_{3} \\
\mathbf{0}_{3} & \mathbf{0}_{3} & \mathbf{I}_{3} & \mathbf{0}_{3} \\
\mathbf{0}_{3} & \mathbf{0}_{3} & \mathbf{0}_{3} & \mathbf{I}_{3} \\
\end{array} \right],
\end{equation}

\begin{equation}
    \mathbf{w} =  \left[
\begin{array}{ccc}
\mathbf{w}_g^T \quad \mathbf{w}_a^T \quad \mathbf{w}_{bg}^T \quad \mathbf{w}_{ba}^T
\end{array}\right]^T. \\ 
\end{equation}
Here, $\mathbf{0}_n$ denotes the $n \times n$ zero matrix, and $\mathbf{I}_n$ denotes the $n \times n$ identity matrix.
\subsection{Lie Group and State Representation in InEKF}
A matrix Lie group $G$ is a smooth manifold consisting of invertible $N \times N$ matrices. Its associated Lie algebra $\mathfrak{g}$ is the tangent space of $G$ at the identity element and has the same dimension as $G$~\cite{barrau2016invariant}. The mappings between the Lie algebra and its vector space representation are given by
\begin{equation}
    (\cdot)^{\wedge} : \mathbb{R}^{\dim(\mathfrak{g})} \rightarrow \mathfrak{g}, \,
    (\cdot)^{\vee} : \mathfrak{g} \rightarrow \mathbb{R}^{\dim(\mathfrak{g})}.
\end{equation}

where $\dim(\mathfrak{g})$ denotes the dimension of the $\mathfrak{g}$. The relationship among matrix Lie groups, Lie algebras, and vector spaces can be expressed as~\cite{sola2018micro}
\begin{equation}
\begin{aligned}
    &\forall \, \boldsymbol{\chi} \in G, \exists \, \boldsymbol{\xi} \in \mathbb{R}^{\dim(\mathfrak{g})}, \boldsymbol{\xi}^{\wedge} \in \mathfrak{g}, \\
    &\text{such that } \boldsymbol{\chi} = \exp_m(\boldsymbol{\xi}^{\wedge}) = \operatorname{Exp}(\boldsymbol{\xi}),
\end{aligned}
\end{equation}
where $\exp_m(\cdot)$ denotes the matrix exponential, and $\operatorname{Exp}(\cdot)$ represents the mapping from the Euclidean space \(\mathbb{R}^{\dim(\mathfrak{g})}\) to the Lie group \(G\).

For all $\boldsymbol{\chi} \in G$ and $\boldsymbol{\xi}^{\wedge} \in \mathfrak{g}$, there exists an associated adjoint matrix $\mathbf{Ad}_{\boldsymbol{\chi}}$~\cite{sola2018micro}, such that
\begin{align}
    \label{eq:Ad}
    (\mathbf{Ad}_{\boldsymbol{\chi}} \boldsymbol{\xi})^{\wedge} &= \boldsymbol{\chi} \boldsymbol{\xi}^{\wedge} \boldsymbol{\chi}^{-1}, \\
    \mathbf{Ad}_{\boldsymbol{\chi}} \mathbf{Ad}_{\boldsymbol{\chi}}^{-1} &= \mathbf{I}_N,
\end{align}
where the adjoint representation $\mathbf{Ad}_{\boldsymbol{\chi}}$ maps elements of the Lie algebra $\mathfrak{g}$ under the group action of $\boldsymbol{\chi} \in G$.

To preserve the group-affine property of the dynamic process model~\cite{barrau2020mathematical, luo2021se_2}, Barrau \textit{et al.} introduced the modified velocity $\mathbf{v}_{ib}^e$ and position $\mathbf{r}_{ib}^e$, which are actually the velocity and position of the body frame with respect to the $i$ frame, expressed in the $e$ frame. The corresponding INS kinematic model in the $e$ frame is given by~\cite{chang2022log}
\begin{equation}
\begin{aligned}
    \label{eq:diff_affine_group}
    \dot{\mathbf{C}}_b^e &= \mathbf{C}_b^e (\boldsymbol{\omega}_{ib}^b \times) - (\boldsymbol{\omega}_{ie}^e \times) \mathbf{C}_b^e, \\
    \dot{\mathbf{v}}_{ib}^e &= \mathbf{C}_b^e \mathbf{f}_{ib}^b - (\boldsymbol{\omega}_{ie}^e \times) \mathbf{v}_{ib}^e + \mathbf{G}_{ib}^e, \\
    \dot{\mathbf{r}}_{ib}^e &= \mathbf{v}_{ib}^e - (\boldsymbol{\omega}_{ie}^e \times) \mathbf{r}_{ib}^e, \\
    \dot{\mathbf{b}}_g &= \mathbf{w}_{bg}, \\
    \dot{\mathbf{b}}_a &= \mathbf{w}_{ba}.
\end{aligned}
\end{equation}
The gravitational acceleration $\mathbf{G}_{ib}^e$ is related to the local gravity vector $\mathbf{g}^e$ by
\begin{equation}
    \mathbf{G}_{ib}^e = \mathbf{g}^e + (\boldsymbol{\omega}_{ie}^e \times)^2 \mathbf{r}_{eb}^e.
\end{equation}
For simplicity, the velocities $\mathbf{v}_{eb}^e$ and $\mathbf{v}_{ib}^e$ are respectively denoted by $\mathbf{v}^e$ and $\overline{\mathbf{v}}^e$. Since $\mathbf{r}_{eb}^e$ and $\mathbf{r}_{ib}^e$ represent the same physical quantity, they are both denoted by $\mathbf{r}^e$~\cite{groves2015principles}. Furthermore, the skew-symmetric matrices $(\boldsymbol{\omega}_{ie}^e \times)$ and $(\boldsymbol{\omega}_{ib}^b \times)$ are abbreviated as $\boldsymbol{\Omega}_{ie}^e$ and $\boldsymbol{\Omega}_{ib}^b$, respectively.

Neglecting the augmented states associated with sensor biases, the rotation matrix $\mathbf{C}_b^e$, the modified velocity $\overline{\mathbf{v}}^e$, and the position $\mathbf{r}^e$ can be compactly represented on the Lie group $\mathbb{SE}_2(3)$~\cite{barrau2020mathematical, chang2022log} as
\begin{equation}
    \boldsymbol{\chi} = 
    \begin{bmatrix}
        \mathbf{C}_b^e & \overline{\mathbf{v}}^e & \mathbf{r}^e \\
        \mathbf{0}_{1 \times 3} & 1 & 0 \\
        \mathbf{0}_{1 \times 3} & 0 & 1
    \end{bmatrix},
\end{equation}
where $\boldsymbol{\chi} \in \mathbb{SE}_2(3)$ denotes the navigation state in the InEKF formulation. Based on the dynamics given in~\eqref{eq:diff_affine_group}, the time derivative of $\boldsymbol{\chi}$ can be expressed as~\cite{chang2022log, tang2022invariant}
\begin{equation}
    \label{eq:chi_diff}
    \begin{aligned}
        &\dot{\boldsymbol{\chi}} = f_u(\boldsymbol{\chi}) \\
        &= 
        \begin{bmatrix}
            \mathbf{C}_b^e \boldsymbol{\Omega}_{ib}^b - \boldsymbol{\Omega}_{ie}^e \mathbf{C}_b^e & \mathbf{C}_b^e \mathbf{f}_{ib}^b - \boldsymbol{\Omega}_{ie}^e \overline{\mathbf{v}}^e + \mathbf{G}_{ib}^e & \mathbf{v}_{ib}^e - \boldsymbol{\Omega}_{ie}^e \mathbf{r}^e \\
            \mathbf{0}_{1\times 3} & 0 & 0 \\
            \mathbf{0}_{1\times 3} & 0 & 0
        \end{bmatrix} \\
        &= \boldsymbol{\chi}
        \begin{bmatrix}
            \boldsymbol{\Omega}_{ib}^b & \mathbf{f}_{ib}^b & \mathbf{0}_{3\times1} \\
            \mathbf{0}_{1\times3} & 0 & 1 \\
            \mathbf{0}_{1\times3} & 0 & 0
        \end{bmatrix}
        +
        \begin{bmatrix}
            -\boldsymbol{\Omega}_{ie}^e & \mathbf{G}_{ib}^e & \mathbf{0}_{3\times1} \\
            \mathbf{0}_{1\times3} & 0 & -1 \\
            \mathbf{0}_{1\times3} & 0 & 0
        \end{bmatrix}
        \boldsymbol{\chi} \\
        &= \boldsymbol{\chi} \mathbf{W} + \mathbf{U} \boldsymbol{\chi},
    \end{aligned}
\end{equation}
where $\mathbf{W}$ and $\mathbf{U}$ are independent of the estimated system states under the assumption that the gravitational acceleration $\mathbf{G}_{ib}^e$ is invariant with respect to position, i.e., variations due to motion are neglected.

It has been shown in~\cite{barrau2016invariant, chang2022log} that for any $\boldsymbol{\chi}_1, \boldsymbol{\chi}_2 \in \mathbb{SE}_2(3)$, both of which are solutions to~\eqref{eq:chi_diff}, the system satisfies the group-affine property
\begin{equation}
    \label{eq:group_affine}
    f_u(\boldsymbol{\chi}_1 \boldsymbol{\chi}_2) = f_u(\boldsymbol{\chi}_1)\boldsymbol{\chi}_2 + \boldsymbol{\chi}_1 f_u(\boldsymbol{\chi}_2) - \boldsymbol{\chi}_1 f_u(\mathbf{I}_d)\boldsymbol{\chi}_2,
\end{equation}
where $\mathbf{I}_d \in \mathbb{SE}_2(3)$ denotes the identity element, with identity rotation and zero translation and velocity components. This condition assures that the system matrices corresponding to the left- and right-invariant errors are independent of the system trajectory~\cite{barrau2016invariant}.
\subsection{Linearization of the InEKF Model}
The left-invariant error of L-InEKF is defined as
\begin{equation}
    \label{eq:eta_l_definition_E}
    \boldsymbol{\eta}_l = \hat{\boldsymbol{\chi}}^{-1} \boldsymbol{\chi} = 
    \begin{bmatrix}
        \hat{\mathbf{C}}_e^b \mathbf{C}_b^e & \hat{\mathbf{C}}_e^b ( \overline{\mathbf{v}}^e - \hat{\overline{\mathbf{v}}}^e ) & \hat{\mathbf{C}}_e^b ( \mathbf{r}^e - \hat{\mathbf{r}}^e ) \\
        \mathbf{0}_{1\times3} & 1 & 0 \\
        \mathbf{0}_{1\times3} & 0 & 1
    \end{bmatrix},
\end{equation}
where $\hat{\boldsymbol{\chi}}$ and $\boldsymbol{\chi}$ denote the estimated and true states on the Lie group $\mathbb{SE}_2(3)$, respectively.

The attitude, velocity, and position errors in the left-invariant formulation are defined respectively as
\begin{align}
    \label{eq:xi_l_definition_E}
    \hat{\mathbf{C}}_e^b \mathbf{C}_b^e &\approx \mathbf{I}_3 + \boldsymbol{\varphi}_l\times, \\
    \mathbf{d}\overline{\mathbf{v}}_l &= \hat{\mathbf{C}}_e^b ( \overline{\mathbf{v}}^e - \hat{\overline{\mathbf{v}}}^e ) = -\hat{\mathbf{C}}_e^b \, \delta\overline{\mathbf{v}}^e, \\
    \mathbf{d}\mathbf{r}_l &= \hat{\mathbf{C}}_e^b ( \mathbf{r}^e - \hat{\mathbf{r}}^e ) = -\hat{\mathbf{C}}_e^b \, \delta\mathbf{r}^e.
\end{align}
Under the small-error assumption, the invariant error $\boldsymbol{\eta}_l$ can be approximated by
\begin{equation}
    \boldsymbol{\eta}_l \approx \mathbf{I}_5 + \boldsymbol{\xi}_l^{\wedge},
\end{equation}
where $\boldsymbol{\xi}_l^{\wedge} \in \mathfrak{se}_2(3)$ is the left-invariant error expressed in the Lie algebra,
\begin{equation}
\boldsymbol{\xi}_{l}=\left[\mathbf{\varphi}_l^T \quad \mathbf{d}\overline{\mathbf{v}}_l^T \quad \mathbf{dr}_l^T \right]^{\mathbf{T}}. \\
\end{equation}
To incorporate the IMU gyroscope and accelerometer biases into the left-invariant error, the augmented left-invariant error state is defined as
\begin{equation}
    \boldsymbol{\xi}_{l(\mathrm{all})} = 
    \begin{bmatrix}
        \boldsymbol{\xi}_l^T & \delta \mathbf{b}_g^T & \delta \mathbf{b}_a^T
    \end{bmatrix}^T.
\end{equation}
The corresponding linearized error dynamics of the augmented system are then given by
\begin{equation}
    \dot{\boldsymbol{\xi}}_{l(\mathrm{all})} = \mathbf{F}_l \boldsymbol{\xi}_{l(\mathrm{all})} + \mathbf{G}_l \mathbf{w},
\end{equation}
where $\mathbf{F}_l$ is the system Jacobian matrix, $\mathbf{G}_l$ is the driving matrix,
\begin{equation}
\mathbf{F}_l = \left[\begin{array}{ccccc}
-\hat{\boldsymbol{\Omega}}_{ib}^b & \mathbf{0}_{3} & \mathbf{0}_{3} & -\mathbf{I}_{3} & \mathbf{0}_{3} \\
-\hat{\mathbf{f}}_{ib}^b \times & -\hat{\boldsymbol{\Omega}}_{ib}^b & \mathbf{0}_{3} & \mathbf{0}_{3} & -\mathbf{I}_{3} \\
\mathbf{0}_{3} & \mathbf{I}_{3} & -\hat{\boldsymbol{\Omega}}_{ib}^b & \mathbf{0}_{3} & \mathbf{0}_{3} \\
\mathbf{0}_{3} & \mathbf{0}_{3} & \mathbf{0}_{3} & \mathbf{0}_{3} & \mathbf{0}_{3} \\
\mathbf{0}_{3} & \mathbf{0}_{3} & \mathbf{0}_{3} & \mathbf{0}_{3} & \mathbf{0}_{3} \\
\end{array}\right],
\end{equation}
\begin{equation}
\mathbf{G}_l = \left[\begin{array}{ccccc}
-\mathbf{I}_{3} & \mathbf{0}_{3} & \mathbf{0}_{3} & \mathbf{0}_{3} \\
\mathbf{0}_{3} & -\mathbf{I}_{3} & \mathbf{0}_{3} & \mathbf{0}_{3} \\
\mathbf{0}_{3} & \mathbf{0}_{3} & \mathbf{0}_{3} & \mathbf{0}_{3} \\
\mathbf{0}_{3} & \mathbf{0}_{3} & \mathbf{I}_{3} & \mathbf{0}_{3} \\
\mathbf{0}_{3} & \mathbf{0}_{3} & \mathbf{0}_{3} & \mathbf{I}_{3} \\
\end{array}\right].
\end{equation}

The right-invariant error of R-InEKF is defined as
\begin{equation}
    \label{eq:eta_r_definition_E}
    \boldsymbol{\eta}_r = \boldsymbol{\chi} \hat{\boldsymbol{\chi}}^{-1} = 
    \begin{bmatrix}
        \mathbf{C}_b^e \hat{\mathbf{C}}_e^b & \overline{\mathbf{v}}^e - \mathbf{C}_b^e \hat{\mathbf{C}}_e^b \hat{\overline{\mathbf{v}}}^e & \mathbf{r}^e - \mathbf{C}_b^e \hat{\mathbf{C}}_e^b \hat{\mathbf{r}}^e \\
        \mathbf{0}_{1 \times 3} & 1 & 0 \\
        \mathbf{0}_{1 \times 3} & 0 & 1
    \end{bmatrix}.
\end{equation}
Let $\boldsymbol{\varphi}_r$, $\mathbf{d}\overline{\mathbf{v}}_r$, and $\mathbf{d}\mathbf{r}_r$ denote the attitude, velocity, and position errors, respectively, in the R-InEKF. These errors are defined as
\begin{align}
    \label{eq:xi_r_definition_E}
    \mathbf{C}_b^e \hat{\mathbf{C}}_e^b &\approx \mathbf{I}_{3} + \boldsymbol{\varphi}_r\times, \\
    \mathbf{d}\overline{\mathbf{v}}_r &= \overline{\mathbf{v}}^e - \mathbf{C}_b^e \hat{\mathbf{C}}_e^b \hat{\overline{\mathbf{v}}}^e = -\delta \overline{\mathbf{v}}^e - (\boldsymbol{\varphi}_r)\times \hat{\overline{\mathbf{v}}}^e, \\
    \mathbf{d}\mathbf{r}_r &= \mathbf{r}^e - \mathbf{C}_b^e \hat{\mathbf{C}}_e^b \hat{\mathbf{r}}^e = -\delta \mathbf{r}^e - (\boldsymbol{\varphi}_r)\times \hat{\mathbf{r}}^e.
\end{align}
Assuming errors are small, and the group error $\boldsymbol{\eta}_r$ can be approximated as
\begin{equation}
    \boldsymbol{\eta}_r \approx \mathbf{I}_5 + \boldsymbol{\xi}_r^{\wedge},
\end{equation}
where the right-invariant error vector is given by
\begin{equation}
    \boldsymbol{\xi}_r =
    \begin{bmatrix}
        \boldsymbol{\varphi}_r^T & \mathbf{d}\overline{\mathbf{v}}_r^T & \mathbf{d}\mathbf{r}_r^T
    \end{bmatrix}^T.
\end{equation}
The augmented right-invariant error vector is defined as
\begin{equation}
    \boldsymbol{\xi}_{r(\mathrm{all})} =
    \begin{bmatrix}
        \boldsymbol{\xi}_r^T & \delta \mathbf{b}_g^T & \delta \mathbf{b}_a^T
    \end{bmatrix}^T.
\end{equation}
The corresponding linearized continuous-time error dynamics for the R-InEKF are given by
\begin{equation}
    \dot{\boldsymbol{\xi}}_{r(\mathrm{all})} = \mathbf{F}_r \boldsymbol{\xi}_{r(\mathrm{all})} + \mathbf{G}_r \mathbf{w},
\end{equation}
where the state transition matrix $\mathbf{F}_r$ and noise driving matrix $\mathbf{G}_r$ are given by
\begin{equation}
    \mathbf{F}_r =
    \begin{bmatrix}
        -\boldsymbol{\Omega}_{ie}^e & \mathbf{0}_{3} & \mathbf{0}_{3} & -\hat{\mathbf{C}}_b^e & \mathbf{0}_{3} \\
        \left(\mathbf{G}_{ib}^e\right)\times & -\boldsymbol{\Omega}_{ie}^e & \mathbf{0}_{3} & -\hat{\overline{\mathbf{v}}}^e\times \hat{\mathbf{C}}_b^e & -\hat{\mathbf{C}}_b^e \\
        \mathbf{0}_{3} & \mathbf{I}_{3} & -\boldsymbol{\Omega}_{ie}^e & -\hat{\mathbf{r}}^e\times \hat{\mathbf{C}}_b^e & \mathbf{0}_{3} \\
        \mathbf{0}_{3} & \mathbf{0}_{3} & \mathbf{0}_{3} & \mathbf{0}_{3} & \mathbf{0}_{3} \\
        \mathbf{0}_{3} & \mathbf{0}_{3} & \mathbf{0}_{3} & \mathbf{0}_{3} & \mathbf{0}_{3}
    \end{bmatrix},
\end{equation}
\begin{equation}
    \mathbf{G}_r =
    \begin{bmatrix}
        -\hat{\mathbf{C}}_b^e & \mathbf{0}_{3} & \mathbf{0}_{3} & \mathbf{0}_{3} \\
        -\hat{\overline{\mathbf{v}}}^e\times \hat{\mathbf{C}}_b^e & -\hat{\mathbf{C}}_b^e & \mathbf{0}_{3} & \mathbf{0}_{3} \\
        -\hat{\mathbf{r}}^e\times \hat{\mathbf{C}}_b^e & \mathbf{0}_{3} & \mathbf{0}_{3} & \mathbf{0}_{3} \\
        \mathbf{0}_{3} & \mathbf{0}_{3} & \mathbf{I}_{3} & \mathbf{0}_{3} \\
        \mathbf{0}_{3} & \mathbf{0}_{3} & \mathbf{0}_{3} & \mathbf{I}_{3}
    \end{bmatrix}.
\end{equation}

It is worth noting that when the bias terms are not included (i.e., without state augmentation), the Jacobian matrices $\mathbf{F}_l$ and $\mathbf{F}_r$ become independent of the current state estimate, and InEKF with additive biases are named as "imperfect" InEKF in~\cite{hartley2020contact}.
\subsection{Observation Model} \label{sec:observation_matrix}
To mitigate the error accumulation of INS, GNSS and ODO are commonly used as auxiliary sensors for integrated navigation. GNSS provides global observations of position and velocity, which are particularly compatible with left-invariant error models. When the GNSS velocity observation is employed, the observation model is formulated as follows
\begin{equation}
\begin{aligned}
    \tilde{\mathbf{y}}_{\text{GNSS(vel)}} &= \mathbf{v}_{\text{GNSS}}^e + \mathbf{n}_{\text{GNSS(vel)}}, \\
    \quad \mathbf{n}_{\text{GNSS(vel)}} &\sim \mathcal{N}(\mathbf{0}_{3 \times 1}, R_{\text{GNSS(vel)}}),
\end{aligned}
\end{equation}
where $\mathbf{v}_{\text{GNSS}}^e$ denotes the velocity measured by GNSS in the $e$ frame, and $\mathbf{n}_{\text{GNSS(vel)}}$ represents zero-mean Gaussian noise with covariance $R_{\text{GNSS(vel)}}$.

The innovation is defined as the difference between the state estimate predicted by the INS and the observation provided by the external sensor, i.e.,
\begin{equation}
    \delta \mathbf{z} = \hat{\mathbf{y}}_{\text{INS}} - \tilde{\mathbf{y}}_{\text{GNSS}}.
\end{equation}

The corresponding observation matrices for the standard EKF, L-InEKF, and R-InEKF are given by
\begin{equation}
\begin{aligned}
    \mathbf{H}_{\text{ekf(vel)}} &= 
    \begin{bmatrix}
        \mathbf{0}_3 & \mathbf{I}_3 & \mathbf{0}_3 & \mathbf{0}_3 & \mathbf{0}_3
    \end{bmatrix}, \\
    \mathbf{H}_{l(\text{vel})} &= 
    \begin{bmatrix}
        \mathbf{0}_3 & -\hat{\mathbf{C}}_b^e & \boldsymbol{\Omega}_{ie}^e \hat{\mathbf{C}}_b^e & \mathbf{0}_3 & \mathbf{0}_3
    \end{bmatrix}, \\
    \mathbf{H}_{r(\text{vel})} &= 
    \begin{bmatrix}
        (\hat{\overline{\mathbf{v}}}^e)\times - (\boldsymbol{\Omega}_{ie}^e \hat{\mathbf{r}}^e)\times & -\mathbf{I}_3 & \boldsymbol{\Omega}_{ie}^e & \mathbf{0}_3 & \mathbf{0}_3
    \end{bmatrix}.
\end{aligned}
\end{equation}

ODO observations are typically the velocity of non-steering wheels, which are aligned with the forward axis of the body frame of the vehicle. Assuming the lateral and vertical velocity of the vehicle are negligible, the ODO observation model can be expressed as
\begin{equation}
    \tilde{\mathbf{y}}_{\text{ODO}} = \mathbf{v}_{eb}^b + \mathbf{n}_{\text{ODO}}, \,
    \mathbf{n}_{\text{ODO}} \sim \mathcal{N}(\mathbf{0}_{3 \times 1}, R_{\text{ODO}}),
\end{equation}
where $\mathbf{n}_{\text{ODO}}$ represents zero-mean Gaussian noise.


The corresponding observation matrices for the EKF, L-InEKF, and R-InEKF are respectively given as
\begin{equation}
\begin{aligned}
    &\mathbf{H}_{\text{ekf(odo)}} = 
    \begin{bmatrix}
        \hat{\mathbf{C}}_e^b \hat{\mathbf{v}}^e \times & \hat{\mathbf{C}}_e^b & \mathbf{0}_3 & \mathbf{0}_3 & \mathbf{0}_3
    \end{bmatrix}, \\
    &\mathbf{H}_{l(\text{odo})} \\ 
    &= \begin{bmatrix}
        (\hat{\mathbf{C}}_e^b (-\hat{\overline{\mathbf{v}}}^e + \boldsymbol{\Omega}_{ie}^e \hat{\mathbf{r}}^e))\times & -\mathbf{I}_3 & \hat{\mathbf{C}}_e^b \boldsymbol{\Omega}_{ie}^e \hat{\mathbf{C}}_b^e & \mathbf{0}_3 & \mathbf{0}_3
    \end{bmatrix}, \\
    &\mathbf{H}_{r(\text{odo})} = 
    \begin{bmatrix}
        -\hat{\mathbf{C}}_e^b (\hat{\mathbf{r}}^e)\times \boldsymbol{\Omega}_{ie}^e & -\hat{\mathbf{C}}_e^b & \hat{\mathbf{C}}_e^b \boldsymbol{\Omega}_{ie}^e & \mathbf{0}_3 & \mathbf{0}_3
    \end{bmatrix}.
\end{aligned}
\end{equation}
\subsection{Relationship Between EKF and InEKF}
This section compares and analyzes the similarities and differences between EKF and InEKF from the perspectives of uncertainty propagation, state update, and covariance transformation.

For systems that satisfy the group-affine property \eqref{eq:group_affine}, the system matrix of the InEKF becomes independent of the trajectory~\cite{barrau2016invariant}. It is generally believed that the InEKF allows for uncertainty propagation without requiring linearization about the current estimated state, which is unavoidable for the EKF~\cite{zhang2024si, luo2023filter}. It is because of this trajectory-independent attribute that InEKF offers improved covariance propagation in such systems.

The relationship between the error states of EKF and those of InEKF, along with the corresponding transformation matrices, is detailed in Appendix~\ref{appx:transformation_relationship}. Specifically,
\begin{equation}
\begin{aligned}
    \label{eq:relation_three_error}
    \boldsymbol{\xi}_l &= \mathbf{J}_l \delta \mathbf{x}, \\
    \boldsymbol{\xi}_r &= \mathbf{J}_r \delta \mathbf{x}, \\
    \boldsymbol{\xi}_r &= \text{Ad}_{\hat{\boldsymbol{\chi}}} \boldsymbol{\xi}_l.
\end{aligned}
\end{equation}
According to Theorem~\ref{thm:prop_equivalence}, if the InEKF and EKF share the same initial system states and equivalent covariance matrices, their error state and covariance propagation results are equivalent. This equivalence holds regardless of whether the system matrices are trajectory-independent, indicating that the propagation processes of filters based on different error state parameterizations are fundamentally equivalent. The explicit forms of the covariance transformation matrices are provided in Appendix~\ref{appx:trans_mat}.



\begin{remark}
     It can be shown that the determinant of the covariance transformation matrix in Appendix~\ref{appx:trans_mat} is equal to one. Hence, the determinant of the covariance matrix in the InEKF is identical to that of the EKF. 
\end{remark}

%
\section{Experimental Results and Analysis}
To verify the equivalence of covariance propagation among different filtering methods and the feasibility of covariance transformation, this section conducts tests using EKF, L-InEKF, and R-InEKF on field test datasets~\cite{zhang2021gnss, wu2013velocity}.
\subsection{Test on Land Vehicle}
The dataset of~\cite{zhang2021gnss} was collected on a small Ackerman-steered vehicle operating on the playground of Wuhan University. The platform is equipped with a consumer-grade IMU, wheel speed odometer (ODO), and GNSS antenna\footnote{https://github.com/i2Nav-WHU/GIOW-release}. A 1000-second segment from $438{,}080\;s$ to $439{,}080\;s$ in the dataset is used for testing, and the vehicle trajectory is shown in Fig.~\ref{fig:traj_car}. The dataset includes raw IMU measurements 200 Hz, real-time RTK positioning results 1 Hz, non-steering wheel speed observations 10 Hz, and reference ground truth 200 Hz for carrier attitude, velocity, and position. Detailed IMU specifications are summarized in Table~\ref{tab:para_field}. A zero-mean Gaussian noise with standard deviation $\mathrm{0.2\; m/s}$ is added to the reference velocity as the velocity observation at 1 Hz. For the ODO observations, the observation noise covariance is set to $\mathrm{0.1\; m/s}$ along all three axes.
\begin{table}[ht!]
\centering
\caption{Technical Specifications of the IMU Used in the Wuhan University Dataset.}
\label{tab:para_field}
\begin{tabular}{ccc}
\hline
Parameter          & Gyroscope & Accelerometer \\ \hline
Random Walk        & $\mathrm{0.15\; deg/\sqrt{h}}$ & $\mathrm{20\; \mu g / \sqrt{Hz}}$ \\
Bias Instability   & $\mathrm{2\; deg/h}$            & $\mathrm{3.6\; \mu g}$             \\ \hline
\end{tabular}
\end{table}
\begin{figure}
    \centering
    \includegraphics[width=0.45\textwidth]{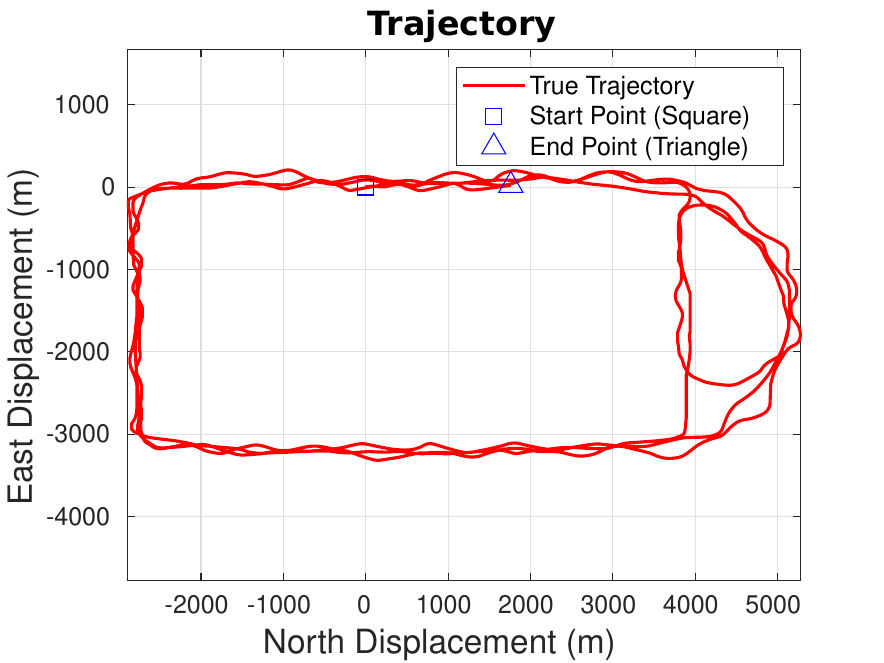}
    \caption{Trajectory of the vehicle during the test period in the Wuhan University Dataset.}
    \label{fig:traj_car}
\end{figure}
\subsubsection{Covariance Propagation Analysis}\label{sec:prop_varify}
This subsection validates the equivalence of covariance propagation in the discrete-time case for practical scenarios with large initial attitude errors: roll and pitch of $60^\circ$ and yaw of $120^\circ$. It is noteworthy that the L-InEKF expresses attitude errors in the body frame; therefore, the attitude covariance values do not correspond directly to the initial attitude error settings. The covariance propagation in the EKF, L-InEKF, and R-InEKF is performed based on their respective covariance definitions. According to Theorem~\ref{thm:prop_equivalence}, there exist transformation relationships among these covariance representations during the propagation process. Therefore, to conduct a meaningful comparison, the covariances must be transformed into a common representation~\cite{chen2024visual}. In this work, the transformation relations in ~\eqref{eq:relation_three_error} are utilized to convert all covariance matrices to the L-InEKF representation for a unified comparison, and results are shown in Fig.~\ref{fig:prop_compare_ct}.
\begin{figure}
    \centering
    \begin{subfigure}[t]{0.5\textwidth}
        \centering
        \includegraphics[width=\textwidth]{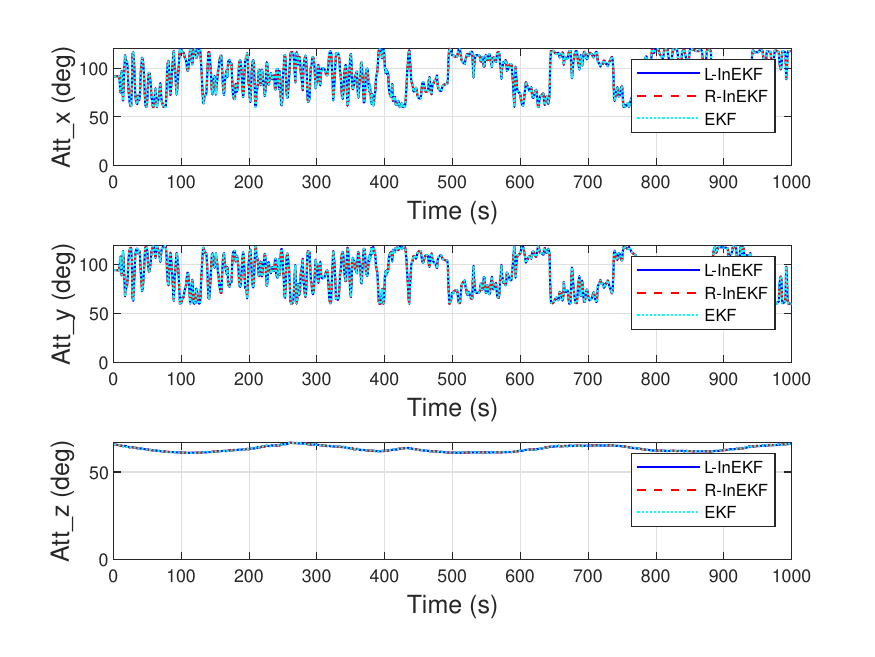}
        \caption{Attitude covariance propagation}
    \end{subfigure}
    \begin{subfigure}[t]{0.5\textwidth}
        \centering
        \includegraphics[width=\textwidth]{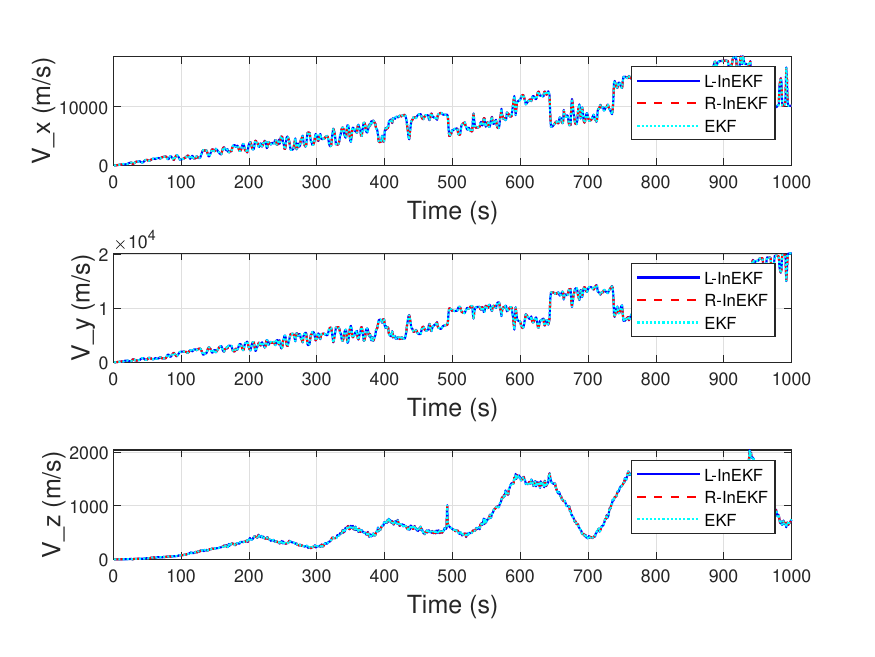}
        \caption{Velocity covariance propagation}
    \end{subfigure}
    
    \begin{subfigure}[t]{0.5\textwidth}
        \centering
        \includegraphics[width=\textwidth]{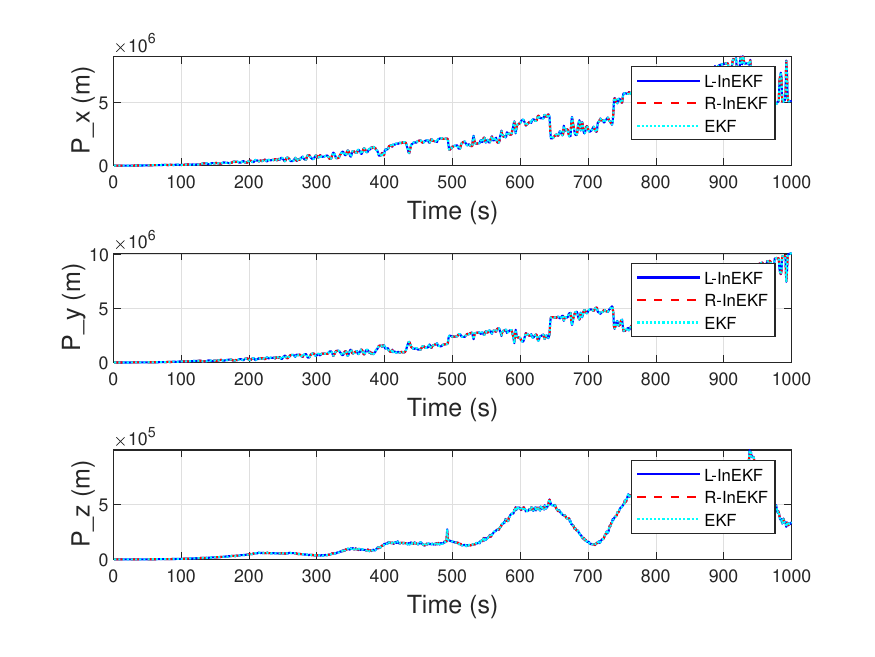}
        \caption{Position covariance propagation}
    \end{subfigure}
    \caption{Covariance propagation comparison among EKF, L-InEKF, and R-InEKF.}
    \label{fig:prop_compare_ct}
\end{figure}

The experimental results demonstrate that, in applications such as INS, where the propagation frequency is high, the discrepancies caused by discretization in different filtering algorithms are negligible. Under such conditions, the covariance propagation processes can be considered equivalent.
\subsubsection{Evaluation of Covariance Transformation}\label{sec:ct_eval}
Our previous studies~\cite{Han2024, han2024covariance} have demonstrated the effectiveness of introducing covariance switch. To further assess covariance transformation's feasibility within the standard EKF, this work adopts EKF as the original estimator and integrates covariance transformation—referred to as CT-EKF—using the coefficient matrices $\mathbf{T}_{\mathrm{ekf} \to \mathrm{l}}$ and $\mathbf{T}_{\mathrm{ekf} \to \mathrm{r}}$ during state estimation. When processing GNSS velocity observations, the transformation matrix $\mathbf{T}_{\mathrm{ekf} \to \mathrm{l}}$ is applied. In this case, the CT-ESKF is expected to yield estimated system states identical to those of the L-InEKF. Similarly, when fusing ODO observations, the transformation matrix $\mathbf{T}_{\mathrm{ekf} \to \mathrm{r}}$ is used.

To highlight the feasibility of covariance transformation under large initial errors, a set of experiments is conducted using GNSS as the sole observation. The initial roll and pitch errors are fixed at $60^\circ$, while the yaw error varies uniformly from $-120^\circ$ to $+120^\circ$ in increments of $5^\circ$. The initial covariance matrix is configured accordingly to reflect these errors. The results of the Monte Carlo experiments are illustrated in Fig.~\ref{fig:gnss_field_rmse_ct}. It can be observed that CT-EKF offers better consistency and estimation performance and coincides with results of the L-InEKF, confirming the feasibility of using covariance transformation matrix $\mathbf{T}_{\mathrm{ekf} \to \mathrm{l}}$ to improve the performance of EKF. These results provide strong evidence for the effectiveness of the proposed covariance transformation strategy as proved in Theorem~\ref{thm:P_relation_ct_s}.
\begin{figure}[h!]
    \centering
    \includegraphics[width=0.5\textwidth]{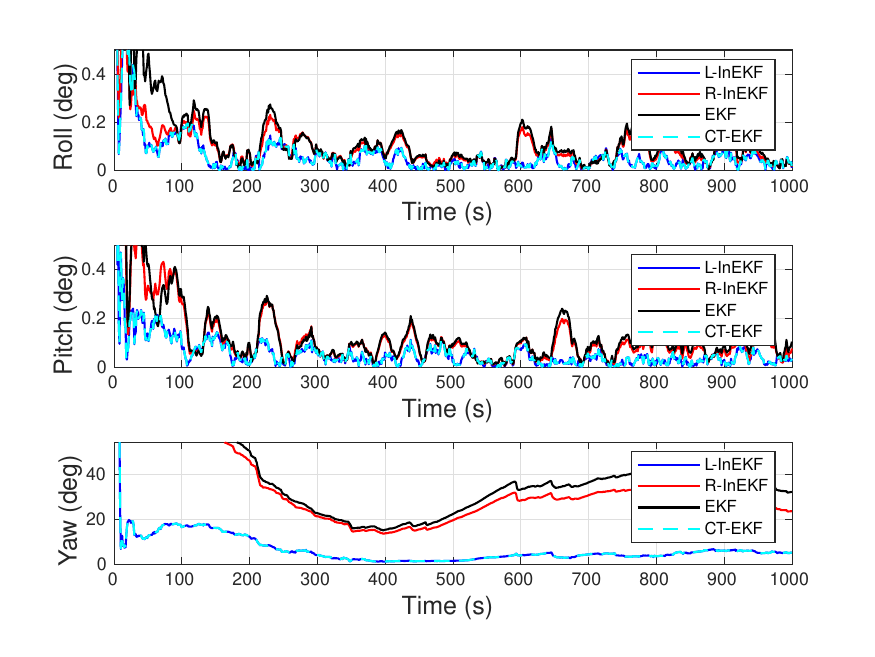}
    \caption{RMSE comparison of attitude estimation errors under GNSS-only observations in land vehicle experiments.}
    \label{fig:gnss_field_rmse_ct}
\end{figure}

Here, we further evaluate the feasibility of covariance transformation under slow covariance propagation. In this condition, slow discrete-time propagation may compromise the equivalence among the covariance matrices of different filtering algorithms, thereby rendering the covariance transformation infeasible. Specifically, the raw IMU data from the dataset are downsampled from 200~Hz to 2~Hz, resulting in a significantly reduced propagation frequency for both the system state and its associated uncertainty.
 GNSS observations are provided at 1 Hz as observations. The initial roll and pitch errors are fixed at $10^\circ$, while the yaw error is uniformly varied between $-120^\circ$ and $+120^\circ$. The RMSE of attitude estimation is shown in Fig.~\ref{fig:GNSS_field_rmse_att_slow}.
\begin{figure}[h!]
    \centering
    \includegraphics[width=0.45\textwidth]{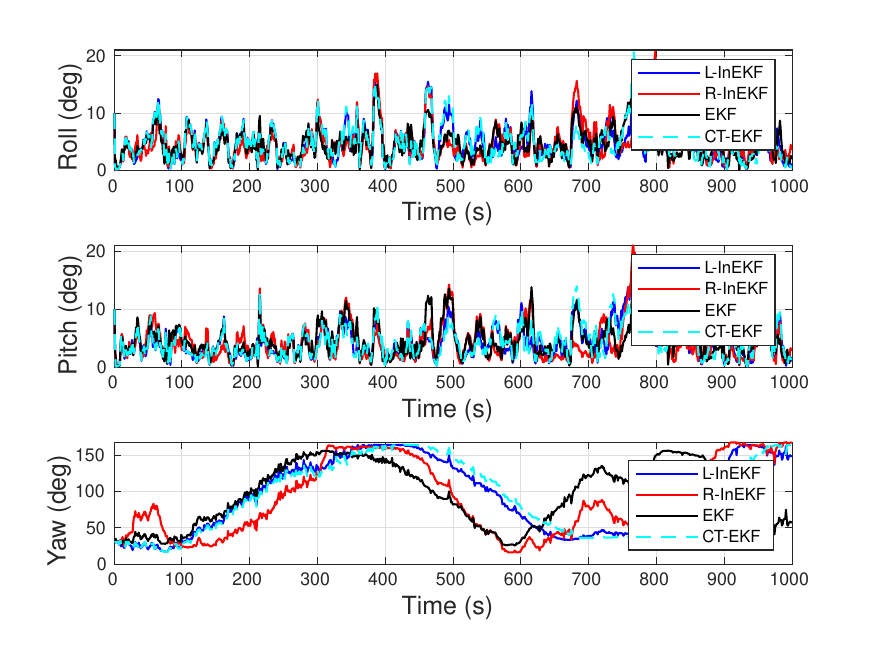}
    \caption{RMSE comparison of attitude estimation errors under GNSS observations when both state and uncertainty are propagated at a low rate.}
    \label{fig:GNSS_field_rmse_att_slow}
\end{figure}

Due to the downsampling of IMU data, the propagated system state no longer accurately reflects the true motion dynamics, resulting in erroneous state and covariance propagation and subsequently degraded estimation performance. However, as shown in the Fig.~\ref{fig:GNSS_field_rmse_att_slow}, the estimation curve of CT-EKF closely aligns with that of L-InEKF. This confirms the feasibility of covariance transformation and indicates that, even under low-rate propagation conditions, the covariance propagation behavior across different filtering algorithms remains nearly equivalent.

In the IMU/ODO integrated navigation scenario, roll and pitch initial errors are fixed at $10^\circ$, while the yaw error varies uniformly from $-60^\circ$ to $+60^\circ$ in increments of $5^\circ$. The initial covariance matrix is set accordingly and results are presented in Fig.~\ref{fig:ODO_field_rmse_ct}. Results indicate that the CT-EKF significantly outperforms the conventional EKF in both estimation accuracy and consistency. Furthermore, the estimation curves of CT-EKF and R-InEKF are perfectly overlapped, which demonstrates that EKF can yield identical results to R-InEKF with the covariance transformation opertation $\mathbf{T}_{\mathrm{ekf} \to \mathrm{r}}$. 

\begin{figure}[h!]
    \centering
    \includegraphics[width=0.45\textwidth]{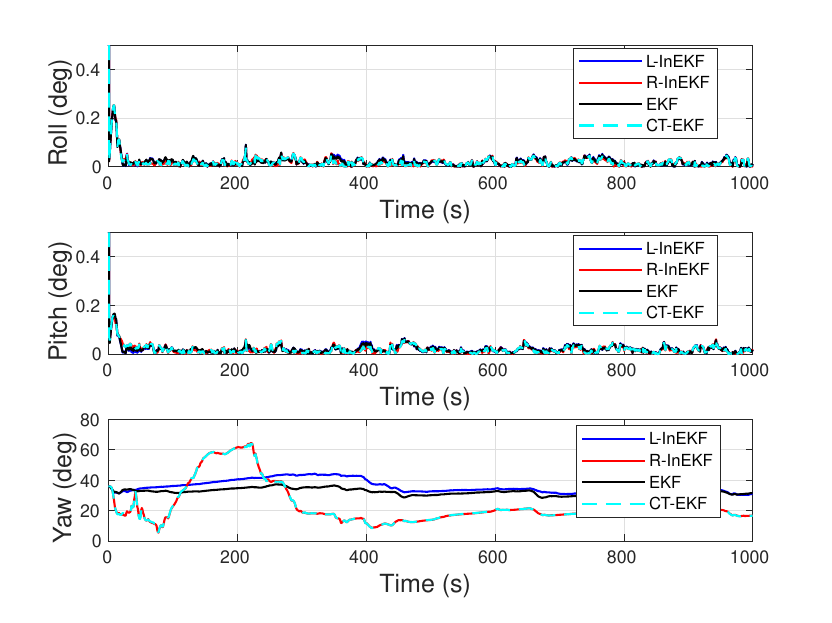}
    \caption{RMSE comparison of attitude estimation errors under ODO-only observations in field experiments.}
    \label{fig:ODO_field_rmse_ct}
\end{figure}

\begin{figure}[h!]
    \centering
    \includegraphics[width=0.45\textwidth]{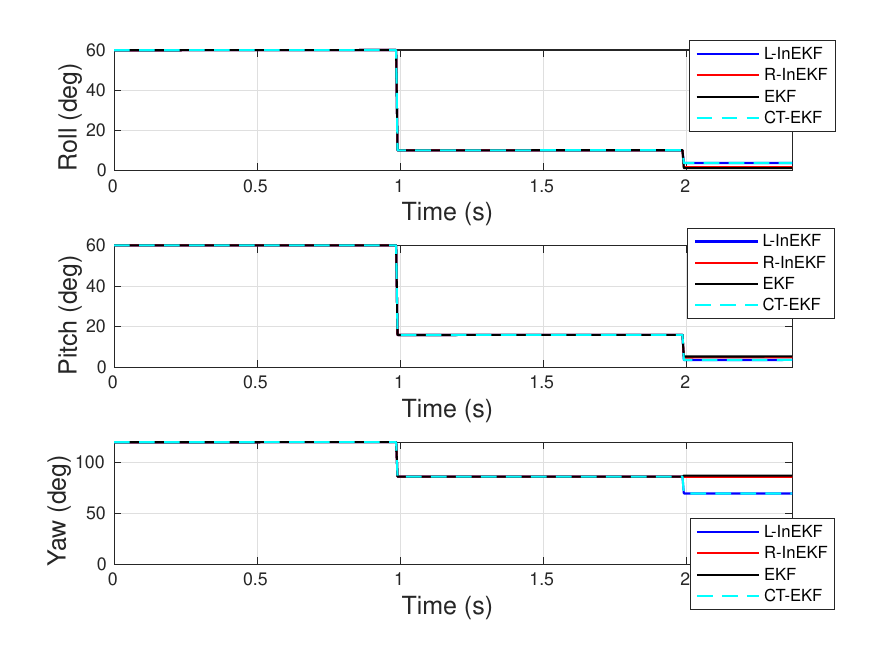}
    \caption{Comparison of first update results for different filtering algorithms.}
    \label{fig:att_snit}
\end{figure}
\begin{figure}[ht!]
    \centering
    \includegraphics[width=0.5\textwidth]{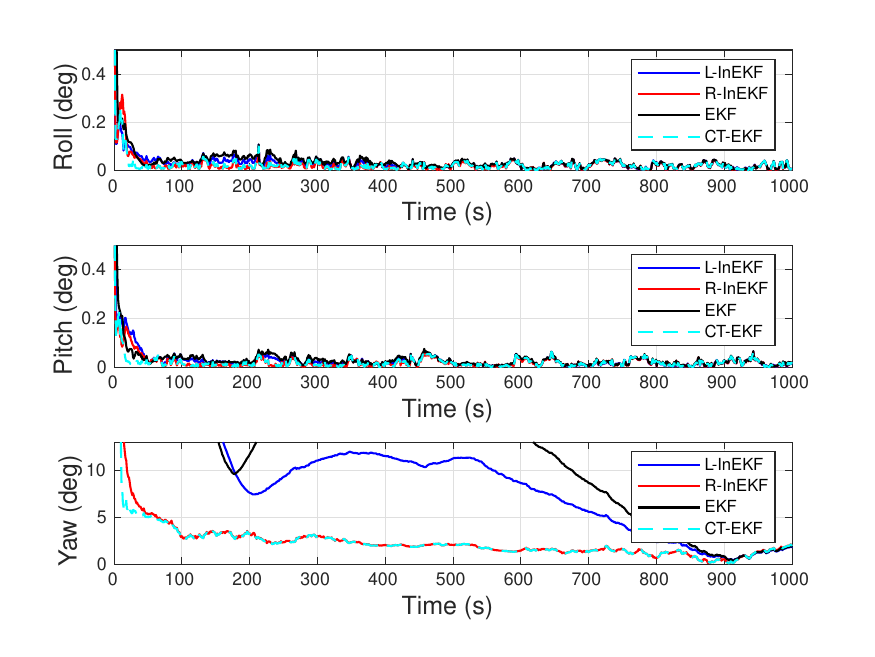}
    \caption{RMSE comparison of attitude estimation errors under GNSS and ODO observations in field experiments.}
    \label{fig:GNSS_ODO_field_rmse_ct}
\end{figure}
\begin{figure*}[htbp]
\centering
\begin{subfigure}[t]{0.45\textwidth}
    \centering
    \includegraphics[width=\textwidth]{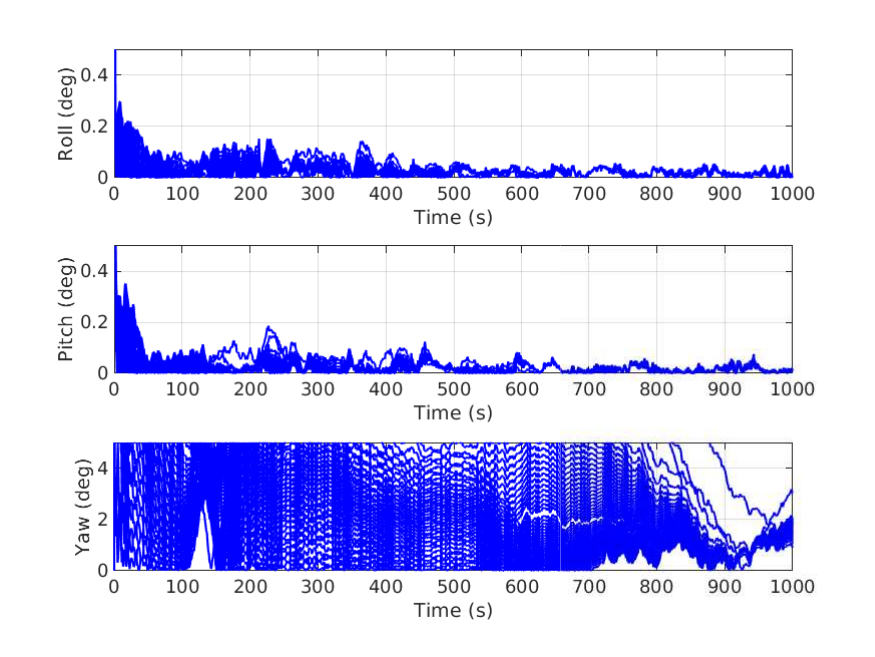}
    \caption{L-InEKF}
\end{subfigure}
\begin{subfigure}[t]{0.45\textwidth}
    \centering
    \includegraphics[width=\textwidth]{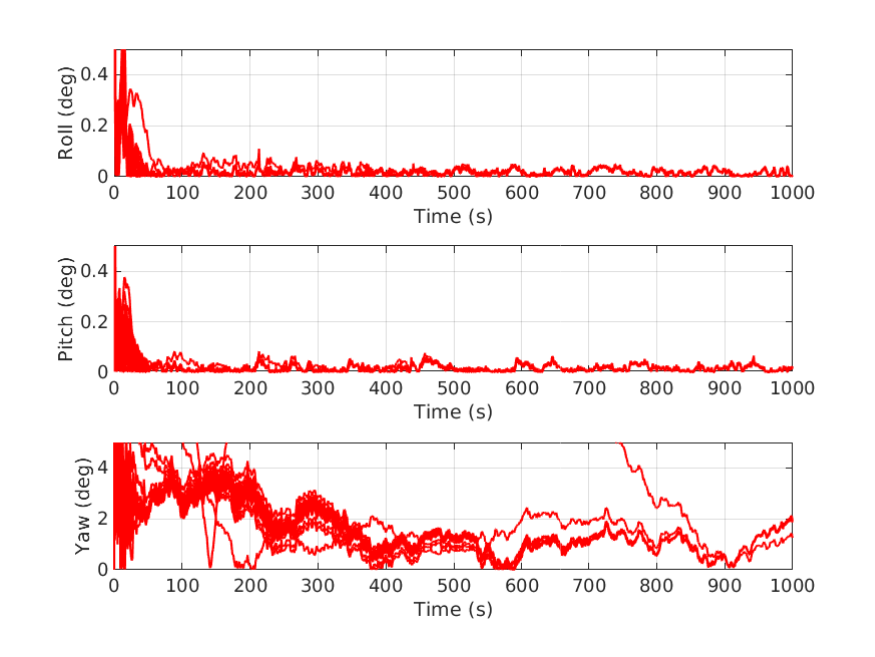}
    \caption{R-InEKF}
\end{subfigure}
\begin{subfigure}[t]{0.45\textwidth}
    \centering
    \includegraphics[width=\textwidth]{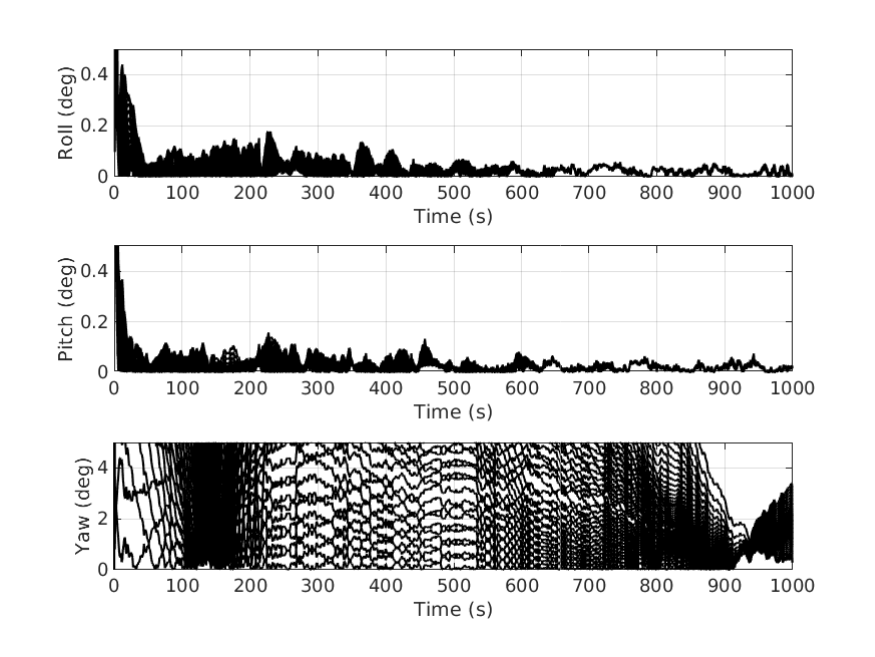}
    \caption{EKF}
\end{subfigure}
\begin{subfigure}[t]{0.45\textwidth}
    \centering
    \includegraphics[width=\textwidth]{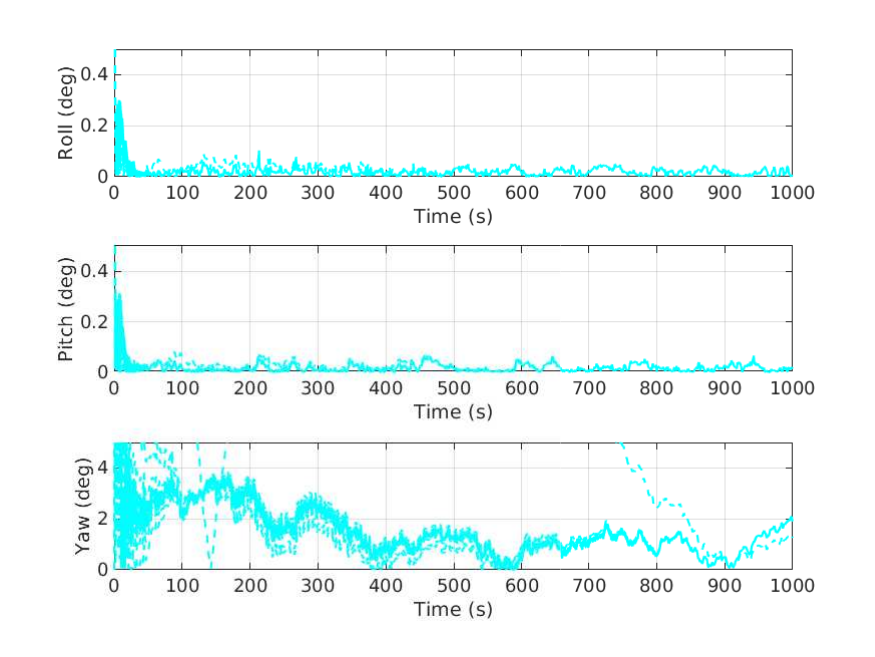}
    \caption{CT-EKF}
\end{subfigure}
\caption{Comparison of attitude estimation errors under GNSS and ODO observations using different methods.}
\label{fig:GNSS_ODO_field_ct}
\end{figure*}
We now focus on the behavior of different filters during their first update under large initial errors. When using the GNSS velocity as observations, the initial roll and pitch errors are fixed at $60^\circ$, and the yaw error is set to $120^\circ$. The estimation errors after the first update are shown in Fig.~\ref{fig:att_snit}, in which all filters perform their first update at $t = 1s$ and the updated system states are identical across all methods. However, as stated in theorem~\ref{pr:result_equ}, the covariance matrices after the first update are not equivalent among the different filters. Consequently, from the second update onward, estimation results begin to diverge among EKF, R-InEKF, L-InEKF, and CT-EKF.

%
\subsubsection{Covariance Transformation in Multi-sensor Navigation}
When there exists a mismatch between the process and observation models in InEKFs, the estimation would be suboptimal and inconsistent~\cite{hartley2019contact, Han2024}. Previous studies show that covariance switch can effectively mitigate the problem~\cite{Han2024}. Here we further investigate the effect of covariance transformation on standard EKF by employing different transformation matrices tailored to specific observation types. For the IMU/GNSS/ODO integrated navigation scenario, CT-EKF utilizes $\mathbf{T}_{\mathrm{ekf} \to \mathrm{l}}$ during GNSS updates and $\mathbf{T}_{\mathrm{ekf} \to \mathrm{r}}$ during ODO updates.

In the Monte Carlo experiments, the initial attitude errors are configured as follows: the roll and pitch errors are set to $60^\circ$, and the yaw error is uniformly distributed from $-150^\circ$ to $+150^\circ$ in increments of $5^\circ$. The initial covariance matrix is configured consistently with these attitude errors. The attitude estimation results are illustrated in Figs.~\ref{fig:GNSS_ODO_field_rmse_ct} and~\ref{fig:GNSS_ODO_field_ct}. It can be observed that CT-EKF significantly outperforms the standard EKF, highlighting the effectiveness of introducing covariance transformation in integrated navigation systems.

Compared to the L-InEKF, CT-EKF demonstrates faster convergence and higher estimation accuracy, indicating that the use of $\mathbf{T}_{\mathrm{ekf} \to \mathrm{r}}$ enables more effective handling of body-frame observations such as ODO. While the performance gain of CT-EKF over R-InEKF is relatively small, this is attributed to the high frequency and accuracy of ODO observations in the dataset, which diminishes the contribution of GNSS observations to overall estimation performance. Nevertheless, CT-EKF still shows improved accuracy in roll and pitch estimation compared to R-InEKF, demonstrating better handling of horizontal attitude components.
\subsection{Test on Aircraft}
To validate the theoretical analysis under high-dynamic conditions, a real-world dataset collected from an aircraft equipped with a high-precision IMU and GNSS receiver is employed~\cite{wu2013velocity}. The corresponding flight trajectory is shown in Fig.~\ref{fig:air_traj_field}. The sampling rates of the IMU and GNSS receiver are 100 Hz and 2 Hz, respectively. The fused result of IMU and RTK-GNSS is used as the reference ground truth. The hardware specifications of the IMU are summarized in Table~\ref{tab:para_field_air}. In this experiment, GNSS velocity observations are used as observations, with the observation noise standard deviation set to $0.1\; \mathrm{m/s}$.
\begin{table}[ht!]
\centering
\caption{Specifications of the IMU used in the aircraft dataset.}
\label{tab:para_field_air}
\begin{tabular}{ccc}
\hline
Parameter             & Gyroscope & Accelerometer \\ \hline
Random Walk           & $\mathrm{0.001\;deg/ \sqrt{h}}$     & $\mathrm{5\; \mu g / \sqrt{Hz}}$            \\
\hline
\end{tabular}
\end{table}

\begin{figure}[ht!]
    \centering
    \includegraphics[width=0.45\textwidth]{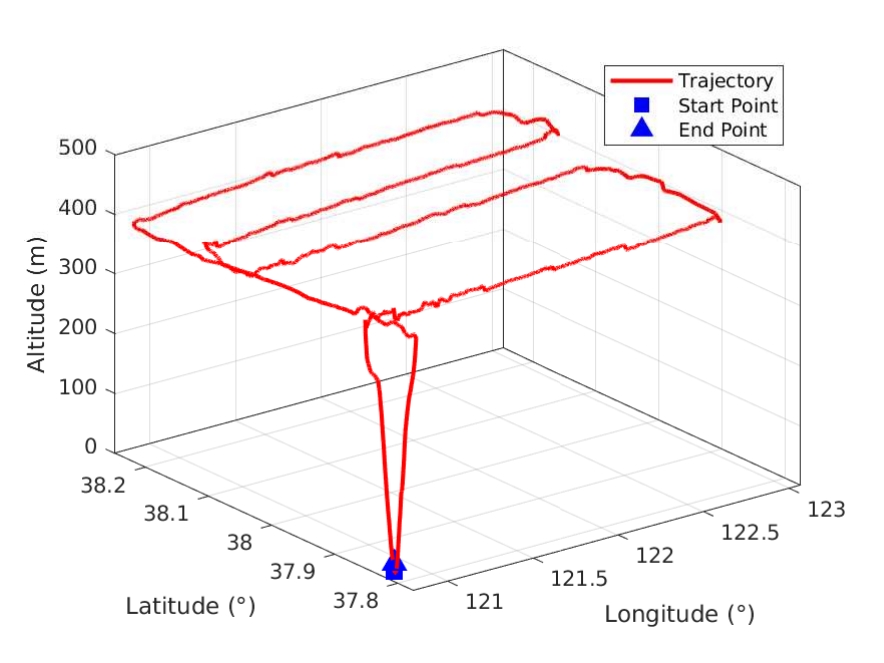}
    \caption{The airplane's trajectory in the flight experiment.}
    \label{fig:air_traj_field}
\end{figure}
To validate the theoretical analysis under large initial attitude errors, the initial roll and pitch errors are set to $10^\circ$, while the yaw error varies uniformly from $-120^\circ$ to $+120^\circ$ in $5^\circ$ increments. The initial covariance matrix is configured accordingly. The experimental results are presented in Fig.~\ref{fig:air_field_test}, and the estimation trajectory of CT-EKF perfectly overlaps with that of the L-InEKF, indicating the equivalence of covariance propagation among different filters even under high-dynamic conditions and large initial errors.
\begin{figure}[ht!]
    \centering
    \includegraphics[width=0.45\textwidth]{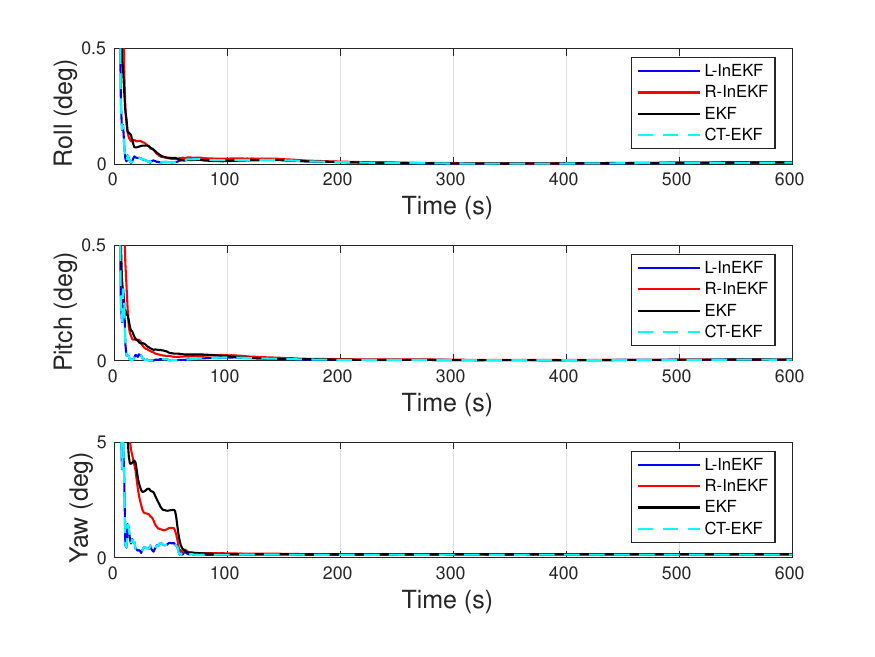}
    \caption{Comparison of attitude estimation RMSE under high-dynamic conditions with GNSS observations.}
    \label{fig:air_field_test}
\end{figure}
\section{Conclusion}
This work introduces a rigorous definition of equivalence between error states and covariance matrices. Based on this foundation, the equivalence of covariance propagation among different filtering algorithms is established, along with the conditions under which this equivalence holds in practical discrete implementations. Theoretical analyses show that the InEKF does not inherently provide superior system uncertainty propagation compared to the classical EKF. Furthermore, the concept of covariance transformation is systematically formulated and distinguished from covariance switch, offering clear theoretical guidance for future algorithm designs. A general theoretical framework termed the CT-ESKF is proposed, unifying various ESKF algorithms and providing a new perspective for designing improved estimators. This is achieved by adjusting the updated covariance to better align with the true system uncertainty. CT-ESKF is applied to fusion scenarios involving both global- and body-frame observations, and experimental results demonstrate that CT-EKF outperforms the L-InEKF, R-InEKF, and classical EKF in terms of estimation accuracy and consistency.

\appendices
\section{Proof of theorem~\ref{thm:prop_equivalence}}\label{appx:proof_of_lemma2}

The error covariance propagation models for $\boldsymbol{\xi}_a$ and $\boldsymbol{\xi}_b$ are given by
\begin{equation}
    \begin{aligned}
        \dot{\mathbf{P}}_a(t) &= \mathbf{F}_a(t) \mathbf{P}_a(t) + \mathbf{P}_a(t) \mathbf{F}_a^T(t) + \mathbf{G}_a(t) \mathbf{Q}(t) \mathbf{G}_a^T(t), \\
        \dot{\mathbf{P}}_b(t) &= \mathbf{F}_b(t) \mathbf{P}_b(t) + \mathbf{P}_b(t) \mathbf{F}_b^T(t) + \mathbf{G}_b(t) \mathbf{Q}(t) \mathbf{G}_b^T(t).
    \end{aligned}
\end{equation}
For any time $t_i \in [t_s, t_f]$, the system uncertainty propagation from $t_s$ to $t_i$ is given by
\begin{equation}
\begin{aligned}
    \mathbf{P}_a(t_i) =& \mathbf{P}_a(t_s) + \int_{t_s}^{t_i} \mathbf{F}_a(t') \mathbf{P}_a(t') \\
    &+ \mathbf{P}_a(t') \mathbf{F}_a^T(t') + \mathbf{G}_a(t') \mathbf{Q}(t') \mathbf{G}_a^T(t') dt',
\end{aligned}
\end{equation}
\begin{equation}
\begin{aligned}
    \mathbf{P}_b(t_i) =& \mathbf{P}_b(t_s) + \int_{t_s}^{t_i} \mathbf{F}_b(t') \mathbf{P}_b(t') \\
     &+ \mathbf{P}_b(t') \mathbf{F}_b^T(t') + \mathbf{G}_b(t') \mathbf{Q}(t') \mathbf{G}_b^T(t') dt'.
\end{aligned}
\end{equation}
By the precise definition of the definite integral, we have
\begin{equation}
    \mathbf{P}_a(t_i) = \mathbf{P}_a(t_s) + \lim_{n \to \infty} \sum_{i=0}^{n-1} \dot{\mathbf{P}}_a\left(t_s + \frac{t_i - t_s}{n} i\right) \frac{t_i - t_s}{n}.
\end{equation}
Let $\tau = \frac{t_i - t_s}{n}$, then within the interval $\tau$, $\mathbf{F}(t)$, $\mathbf{G}(t)$, and $\dot{\mathbf{P}}(t)$ can be considered constants.
Therefore, for any $\tau$→0

\begin{equation}
    \label{eq:prop_one_step_a}
    \begin{aligned}
        \mathbf{P}_a(t_s + \tau) =& \mathbf{P}_a(t_s) + \dot{\mathbf{P}}_a(t_s) \tau \\
        =& \mathbf{P}_a(t_s) + \big(\mathbf{F}_a(t_s) \mathbf{P}_a(t_s) + \mathbf{P}_a(t_s) \mathbf{F}_a^T(t_s) \\
        &+ \mathbf{G}_a(t_s) \mathbf{Q}(t_s) \mathbf{G}_a^T(t_s)\big) \tau,
    \end{aligned}
\end{equation}
\begin{equation}
            \label{eq:prop_all_step_a}
            \begin{aligned}
                \mathbf{P}_a(t_i) &= \mathbf{P}_a(t_s) + \lim_{n \rightarrow \infty} \sum_{i=0}^{n-1} \dot{\mathbf{P}}_a \left(t_s+\tau i\right) \tau \\
                                  &= \mathbf{P}_a(t_s+\tau) + \lim_{n \rightarrow \infty} \sum_{i=1}^{n-1} \dot{\mathbf{P}}_a \left(t_s+\tau i\right) \tau \\
                                  & \vdots \\
                                  &= \lim_{n \rightarrow \infty}(\mathbf{P}_a(t_s+(n-1)\tau) + \sum_{i=n-1}^{n-1} \dot{\mathbf{P}}_a \left(t_s+\tau i\right) \tau).
            \end{aligned}
        \end{equation}
        Similarly, for $\mathbf{P}_b$ and $\mathbf{A}(\hat{\mathbf{x}}(t))$, we have analogous expressions to \eqref{eq:prop_one_step_a} and \eqref{eq:prop_all_step_a}
\begin{equation}
    \mathbf{P}_b(t_i) = \mathbf{P}_b(t_s) + \lim_{n \to \infty} \sum_{i=0}^{n-1} \dot{\mathbf{P}}_b \left(t_s + \tau i \right) \tau,
\end{equation}
\begin{equation}
    \mathbf{A}(\hat{\mathbf{x}}(t_i)) = \mathbf{A}(\hat{\mathbf{x}}(t_s)) + \lim_{n \to \infty} \sum_{i=0}^{n-1} \dot{\mathbf{A}}(\hat{\mathbf{x}}(t_s + \tau i)) \tau.
\end{equation}
Similarly, considering only the propagation over a small time interval $\tau$, we have
\begin{equation}
            \label{eq:prop_one_step_b}
            \begin{aligned}
                \mathbf{P}_b(t_s+\tau) =& \mathbf{P}_b(t_s) + \dot{\mathbf{P}}_b \left(t_s \right) \tau \\
                                       =& \mathbf{P}_b(t_s) + (\mathbf{F}_b(t_s) \mathbf{P}_b(t_s) + \mathbf{P}_b(t_s) \mathbf{F}_b^T(t_s) \\
                                        & + \mathbf{G}_b(t_s) \mathbf{Q}(t_s) \mathbf{G}_b^T(t_s))\tau,
            \end{aligned}
        \end{equation}
        \begin{equation}
            \label{eq:prop_A_one_step}
            \mathbf{A}(\hat{\mathbf{x}}(t_s+\tau)) = \mathbf{A}(\hat{\mathbf{x}}(t_s)) + \dot{\mathbf{A}}(\hat{\mathbf{x}}(t_s)) \tau.
        \end{equation}
        Based on the transformation relationship between the system matrices $\mathbf{F}_a$ and $\mathbf{F}_b$ for different error states in \eqref{eq:F_relation} and the initial condition in \eqref{eq:switch_init}, we obtain
        \begin{equation}
            \label{eq:p_a_derived}
            \begin{aligned}
                &\mathbf{F}_a(t_s) \mathbf{P}_a(t_s) \tau \\
                &= \mathbf{F}_a(t_s) \, \mathbf{A}(\hat{\mathbf{x}}(t_s)) \, \mathbf{P}_b(t_s) \, \mathbf{A}^T(\hat{\mathbf{x}}(t_s)) \tau \\
                    &= (\mathbf{A}(\hat{\mathbf{x}}(t_s)) \, \mathbf{F}_b(t_s) + \dot{\mathbf{A}}(\hat{\mathbf{x}}(t_s))) \, \mathbf{P}_b(t_s) \, \mathbf{A}^T(\hat{\mathbf{x}}(t_s)) \tau \\
                    &= \mathbf{A}(\hat{\mathbf{x}}(t_s)) \, \mathbf{F}_b(t_s) \, \mathbf{P}_b(t_s) \, \mathbf{A}^T(\hat{\mathbf{x}}(t_s)) \tau \\
                    & \quad + \dot{\mathbf{A}}(\hat{\mathbf{x}}(t_s)) \, \mathbf{P}_b(t_s) \, \mathbf{A}^T(\hat{\mathbf{x}}(t_s)) \tau,
            \end{aligned}
        \end{equation}
        then,
        \begin{equation}
            \begin{aligned}
                &\mathbf{P}_a(t_s) \mathbf{F}_a^T(t_s) \tau \\
                &= (\mathbf{F}_a(t_s) \, \mathbf{P}_a(t_s))^T \tau\\
                &=  \mathbf{A}(\hat{\mathbf{x}}(t_s)) \,\mathbf{P}_b(t_s) \, \mathbf{F}_b^T(t_s) \, \mathbf{A}^T(\hat{\mathbf{x}}(t_s)) \tau \\
                & \quad + \mathbf{A}(\hat{\mathbf{x}}(t_s)) \, \mathbf{P}_b(t_s) \, \dot{\mathbf{A}}^T(\hat{\mathbf{x}}(t_s)) \tau.
            \end{aligned}
        \end{equation}
        Based on \eqref{eq:prop_A_one_step}, we have
\begin{equation}
    \label{eq:P1_relation}
    \begin{aligned}
        &\mathbf{A}(\hat{\mathbf{x}}(t_s+\tau)) \mathbf{P}_b(t_s) \mathbf{A}^T(\hat{\mathbf{x}}(t_s+\tau)) \\
        &= \big(\mathbf{A}(\hat{\mathbf{x}}(t_s)) + \dot{\mathbf{A}}(\hat{\mathbf{x}}(t_s)) \tau \big) \mathbf{P}_b(t_s) \big(\mathbf{A}(\hat{\mathbf{x}}(t_s)) + \dot{\mathbf{A}}(\hat{\mathbf{x}}(t_s)) \tau \big)^T \\
        &= \mathbf{A}(\hat{\mathbf{x}}(t_s)) \mathbf{P}_b(t_s) \mathbf{A}^T(\hat{\mathbf{x}}(t_s)) + \mathbf{A}(\hat{\mathbf{x}}(t_s)) \mathbf{P}_b(t_s) \dot{\mathbf{A}}^T(\hat{\mathbf{x}}(t_s)) \tau \\
        &\quad + \dot{\mathbf{A}}(\hat{\mathbf{x}}(t_s)) \mathbf{P}_b(t_s) \mathbf{A}^T(\hat{\mathbf{x}}(t_s)) \tau, \\
        &= \mathbf{P}_a(t_s) + \mathbf{A}(\hat{\mathbf{x}}(t_s)) \mathbf{P}_b(t_s) \dot{\mathbf{A}}^T(\hat{\mathbf{x}}(t_s)) \tau \\
        &\quad + \dot{\mathbf{A}}(\hat{\mathbf{x}}(t_s)) \mathbf{P}_b(t_s) \mathbf{A}^T(\hat{\mathbf{x}}(t_s)) \tau,
    \end{aligned}
\end{equation}
\begin{equation}
            \label{eq:P2_relation}
            \begin{aligned}
                &\mathbf{A}(\hat{\mathbf{x}}(t_s+\tau)) \, \mathbf{F}_b(t_s) \, \mathbf{P}_b(t_s) \, \tau \mathbf{A}^T(\hat{\mathbf{x}}(t_s+\tau)) \\
                & \quad = (\mathbf{A}(\hat{\mathbf{x}}(t_s)) + \dot{\mathbf{A}}(\hat{\mathbf{x}}(t_s)) \tau) \, \mathbf{F}_b(t_s) \, \mathbf{P}_b(t_s) \, \\
                & \quad \cdot (\mathbf{A}(\hat{\mathbf{x}}(t_s)) + \dot{\mathbf{A}}(\hat{\mathbf{x}}(t_s)) \tau)^T \tau \\
                & \quad = \mathbf{A}(\hat{\mathbf{x}}(t_s)) \, \mathbf{F}_b(t_s) \, \mathbf{P}_b(t_s) \, \mathbf{A}^T(\hat{\mathbf{x}}(t_s)) \tau,
            \end{aligned}
        \end{equation}
where higher-order terms in $\tau^2$ are neglected as $\tau \to 0$. Similarly,
\begin{equation}
    \label{eq:P3_relation}
    \begin{aligned}
        &\mathbf{A}(\hat{\mathbf{x}}(t_s+\tau)) \mathbf{P}_b(t_s) \mathbf{F}_b^T(t_s) \, \tau \mathbf{A}^T(\hat{\mathbf{x}}(t_s+\tau)) \\
        &\quad= \mathbf{A}(\hat{\mathbf{x}}(t_s)) \mathbf{P}_b(t_s) \mathbf{F}_b^T(t_s) \mathbf{A}^T(\hat{\mathbf{x}}(t_s)) \tau.
    \end{aligned}
\end{equation}
Combining \eqref{eq:p_a_derived}-\eqref{eq:P3_relation}, we obtain
\begin{equation}
\label{eq:p_part_relation}
\begin{aligned}
    &\mathbf{P}_a(t_s) + \mathbf{F}_a(t_s) \mathbf{P}_a(t_s) \tau+ \mathbf{P}_a(t_s) \mathbf{F}_a^T(t_s) \tau \\
    &=\mathbf{A}(\hat{\mathbf{x}}(t_s+\tau)) (\mathbf{P}_b(t_s) + \mathbf{F}_b(t_s) \, \mathbf{P}_b(t_s) \tau \\
    &\quad + \mathbf{P}_b(t_s) \mathbf{F}_b^T(t_s))\mathbf{A}^T(\hat{\mathbf{x}}(t_s+\tau)).
\end{aligned}
\end{equation}
Based on \eqref{eq:G_relation}, the relationship between the system noise terms is given by
\begin{equation}
    \label{eq:P4_relation}
    \begin{aligned}
        &\mathbf{A}(\hat{\mathbf{x}}(t_s+\tau)) \mathbf{G}_b(t_s) \mathbf{Q}(t_s) \mathbf{G}_b^T(t_s) \mathbf{A}^T(\hat{\mathbf{x}}(t_s+\tau)) \tau \\
        &\quad= \mathbf{A}(\hat{\mathbf{x}}(t_s)) \mathbf{G}_b(t_s) \mathbf{Q}(t_s) \mathbf{G}_b^T(t_s) \mathbf{A}^T(\hat{\mathbf{x}}(t_s)) \tau, \\
        &\quad= \mathbf{G}_a(t_s) \mathbf{Q}(t_s) \mathbf{G}_a^T(t_s) \tau.
    \end{aligned}
\end{equation}
Based on~\eqref{eq:p_part_relation} and \eqref{eq:P4_relation}, we obtain
\begin{equation}
    \begin{aligned}
        \mathbf{A}(\hat{\mathbf{x}}(t_s+\tau)) \mathbf{P}_b(t_s+\tau) \mathbf{A}^T(\hat{\mathbf{x}}(t_s+\tau)) = \mathbf{P}_a(t_s+\tau).
    \end{aligned}
\end{equation}
Repeating the above derivation yields
\begin{equation}
\begin{aligned}
    &\lim_{n \to \infty} \mathbf{P}_a(t_s + i \,\tau) \\
    &= \lim_{n \to \infty} \mathbf{A}(\hat{\mathbf{x}}(t_s + i \tau)) \mathbf{P}_b(t_s + i \,\tau) \mathbf{A}^T(\hat{\mathbf{x}}(t_s + i \tau)), i \in [0, n].
\end{aligned}
\end{equation}
Summing over the interval, we obtain
\begin{equation}
    \begin{aligned}
        \mathbf{P}_a(t_i) &= \lim_{n \to \infty} \mathbf{P}_a(t_s + n \,\tau) \\
        &= \lim_{n \to \infty} \mathbf{A}(\hat{\mathbf{x}}(t_s + n \tau)) \mathbf{P}_b(t_s + n \,\tau) \mathbf{A}^T(\hat{\mathbf{x}}(t_s + n \tau)) \\
        &= \mathbf{A}(\hat{\mathbf{x}}(t_i)) \mathbf{P}_b(t_i) \mathbf{A}^T(\hat{\mathbf{x}}(t_i)).
    \end{aligned}
\end{equation}
The error state propagation from time $t_s$ to $t_i$ can be computed as
\begin{equation}
    \begin{aligned}
        \boldsymbol{\xi}_a(t_i) &= \boldsymbol{\xi}_a(t_s) + \int_{t_s}^{t_i} \mathbf{F}_a(t') \boldsymbol{\xi}_a(t') \, dt' \\
        &= \boldsymbol{\xi}_a(t_s) + \lim_{n \to \infty} \sum_{i=0}^{n-1} \dot{\boldsymbol{\xi}}_a \left(t_s + \frac{t_i - t_s}{n} i \right) \frac{t_i - t_s}{n},
    \end{aligned}
\end{equation}
\begin{equation}
    \begin{aligned}
        \boldsymbol{\xi}_b(t_i) &= \boldsymbol{\xi}_b(t_s) + \int_{t_s}^{t_i} \mathbf{F}_b(t') \boldsymbol{\xi}_b(t') \, dt' \\
        &= \boldsymbol{\xi}_b(t_s) + \lim_{n \to \infty} \sum_{i=0}^{n-1} \dot{\boldsymbol{\xi}}_b \left(t_s + \frac{t_i - t_s}{n} i \right) \frac{t_i - t_s}{n}.
    \end{aligned}
\end{equation}
Similarly, considering only a single propagation step of duration $\tau$,
\begin{equation}
    \boldsymbol{\xi}_a(t_s + \tau) = \boldsymbol{\xi}_a(t_s) + \mathbf{F}_a(t_s) \boldsymbol{\xi}_a(t_s) \tau,
\end{equation}
\begin{equation}
    \boldsymbol{\xi}_b(t_s + \tau) = \boldsymbol{\xi}_b(t_s) + \mathbf{F}_b(t_s) \boldsymbol{\xi}_b(t_s) \tau.
\end{equation}
Based on \eqref{eq:F_relation}, and \eqref{eq:prop_A_one_step}, we have
\begin{equation}
    \begin{aligned}
        &\mathbf{A}(\hat{\mathbf{x}}(t_s+\tau)) \boldsymbol{\xi}_b(t_s + \tau) \\
        &\quad = \left(\mathbf{A}(\hat{\mathbf{x}}(t_s)) + \dot{\mathbf{A}}(\hat{\mathbf{x}}(t_s)) \tau \right) \left(\boldsymbol{\xi}_b(t_s) + \mathbf{F}_b(t_s) \boldsymbol{\xi}_b(t_s) \tau \right) \\
        &\quad = \left(\mathbf{A}(\hat{\mathbf{x}}(t_s)) + \dot{\mathbf{A}}(\hat{\mathbf{x}}(t_s)) \tau \right) \boldsymbol{\xi}_b(t_s) \\
        & \quad  \quad + \mathbf{A}(\hat{\mathbf{x}}(t_s)) \mathbf{F}_b(t_s) \boldsymbol{\xi}_b(t_s) \tau \\
        &\quad = \mathbf{A}(\hat{\mathbf{x}}(t_s)) \boldsymbol{\xi}_b(t_s) + \mathbf{F}_a(t_s) \mathbf{A}(\hat{\mathbf{x}}(t_s)) \boldsymbol{\xi}_b(t_s) \tau \\
        &\quad = \boldsymbol{\xi}_a(t_s) + \mathbf{F}_a(t_s) \boldsymbol{\xi}_a(t_s) \tau \\
        &\quad = \boldsymbol{\xi}_a(t_s + \tau).
    \end{aligned}
\end{equation}
Similarly, we have
\begin{equation}
    \begin{aligned}
        \boldsymbol{\xi}_a(t_i) &= \lim_{n \rightarrow \infty} \boldsymbol{\xi}_a(t_s + n \tau) \\
        &= \lim_{n \rightarrow \infty} \mathbf{A}(\hat{\mathbf{x}}(t_s + n \tau)) \boldsymbol{\xi}_b(t_s + n \tau) \\
        &= \mathbf{A}(\hat{\mathbf{x}}(t_i)) \boldsymbol{\xi}_b(t_i).
    \end{aligned}
\end{equation}
Thus, Theorem~\ref{thm:prop_equivalence} is proved.

\section{Proof of Theorem~\ref{thm:prop_equivalence_disc}}\label{appx:proof_of_theorem6}
Specifically, we consider the single-step propagation from time $t_k$ to $t_{k+1}$.

\begin{theorem}[Equivalence of Error State and Covariance Propagation in Discrete Time]
    \label{thm:prop_equivalence_disc}
    Suppose the transformation matrix $\mathbf{A}(\hat{\mathbf{x}})$ between $\boldsymbol{\xi}_a$ and $\boldsymbol{\xi}_b$, as well as between $\mathbf{P}_a$ and $\mathbf{P}_b$, satisfies the following approximations at the discrete time instants $k$ and $k+1$, with step size $\tau$
    \begin{equation}
        \label{eq:sim1}
        \mathbf{A}(k) + \dot{\mathbf{A}}(k) \tau \approx \mathbf{A}(k+1),
    \end{equation}
    \begin{equation}
        \label{eq:sim2}
        \mathbf{A}(k) \tau \approx (\mathbf{A}(k) + \dot{\mathbf{A}}(k) \tau) \tau \approx \mathbf{A}(k+1) \tau,
    \end{equation}
    then under discrete-time propagation, the relationships between $\boldsymbol{\xi}_a$ and $\boldsymbol{\xi}_b$, as well as between $\mathbf{P}_a$ and $\mathbf{P}_b$, approximately satisfy the equivalence conditions defined in Definitions~\ref{def:error_equ} and~\ref{def:covariance_equ}.
    
    \begin{proof}
    Consider the discrete-time propagation of the covariance matrices $\mathbf{P}_a$ and $\mathbf{P}_b$ over a single time step
    \begin{equation}
        \begin{aligned}
            &\mathbf{P}_a(k+1) \\
            &= (\mathbf{I} + \mathbf{F}_a(k) \tau) \mathbf{P}_a(k) (\mathbf{I} + \mathbf{F}_a(k) \tau)^T + \mathbf{D}_a(k) \mathbf{Q} \mathbf{D}_a(k)^T \tau, \\
            &\mathbf{P}_b(k+1) \\
            &= (\mathbf{I} + \mathbf{F}_b(k) \tau) \mathbf{P}_b(k) (\mathbf{I} + \mathbf{F}_b(k) \tau)^T + \mathbf{D}_b(k) \mathbf{Q} \mathbf{D}_b(k)^T \tau,
        \end{aligned}
    \end{equation}
    where the matrices $\mathbf{F}_a(k)$, $\mathbf{F}_b(k)$, $\mathbf{D}_a(k)$, and $\mathbf{D}_b(k)$ represent the system dynamics and noise input matrices evaluated at time step $k$.
    
    Based on \eqref{eq:sim1} and \eqref{eq:sim2}, the following relations hold
    \begin{equation}
    \begin{aligned}
        &(\mathbf{I} + \mathbf{F}_a(k) \tau) \mathbf{P}_a(k) (\mathbf{I} + \mathbf{F}_a(k) \tau)^T \\
        &= (\mathbf{I} + \mathbf{F}_a(k) \tau) \mathbf{A}(k) \mathbf{P}_b(k) \mathbf{A}(k)^T (\mathbf{I} + \mathbf{F}_a(k) \tau)^T \\
        &= (\mathbf{A}(k) + \mathbf{F}_a(k) \mathbf{A}(k) \tau) \mathbf{P}_b(k) \\
        & \quad \cdot (\mathbf{A}(k) + \mathbf{F}_a(k) \mathbf{A}(k) \tau)^T \\
        &= (\mathbf{A}(k) + (\mathbf{A}(k) \mathbf{F}_b(k) + \dot{\mathbf{A}}(k)) \tau) \mathbf{P}_b(k) \\
        & \quad \cdot(\mathbf{A}(k) + (\mathbf{A}(k) \mathbf{F}_b(k) + \dot{\mathbf{A}}(k)) \tau)^T \\
        &\approx (\mathbf{A}(k+1) + \mathbf{A}(k+1) \mathbf{F}_b(k) \tau) \mathbf{P}_b(k) \\
        & \quad \cdot (\mathbf{A}(k+1) + \mathbf{A}(k+1) \mathbf{F}_b(k) \tau)^T \\
        &\approx \mathbf{A}(k+1) \left[(\mathbf{I} + \mathbf{F}_b(k) \tau) \mathbf{P}_b(k) (\mathbf{I} + \mathbf{F}_b(k) \tau)^T \right] \mathbf{A}(k+1)^T,
    \end{aligned}
    \end{equation}
    \begin{equation}
        \begin{split}
            &\mathbf{D}_a(k) \mathbf{Q} \mathbf{D}_a(k)^T \tau \\
            &= \mathbf{A}(k) \mathbf{D}_b(k) \mathbf{Q} \mathbf{D}_b(k)^T \mathbf{A}(k)^T \tau \\
            &\approx (\mathbf{A}(k) + \dot{\mathbf{A}}(k) \tau) \mathbf{D}_b(k) \mathbf{Q} \mathbf{D}_b(k)^T (\mathbf{A}(k) + \dot{\mathbf{A}}(k) \tau)^T \tau \\
            &\approx \mathbf{A}(k+1) \mathbf{D}_b(k) \mathbf{Q} \mathbf{D}_b(k)^T \mathbf{A}(k+1)^T \tau.
        \end{split}
        \end{equation}
    In summary, the following relation hold
    \begin{equation}
    \mathbf{P}_a(k+1) \approx \mathbf{A}(k+1) \mathbf{P}_b(k+1) \mathbf{A}(k+1)^T.
    \end{equation}
    Similarly,
    \begin{equation}
    \boldsymbol{\xi}_a(k+1) = (\mathbf{I} + \mathbf{F}_a(k) \tau) \boldsymbol{\xi}_a(k),
    \end{equation}
    \begin{equation}
    \boldsymbol{\xi}_b(k+1) = (\mathbf{I} + \mathbf{F}_b(k) \tau) \boldsymbol{\xi}_b(k),
    \end{equation}
    \begin{equation}
    \begin{split}
        &\boldsymbol{\xi}_a(k+1) \\
        &= (\mathbf{I} + \mathbf{F}_a(k) \tau) \mathbf{A}(k) \boldsymbol{\xi}(k) \\
        &= (\mathbf{A}(k) + \mathbf{F}_a(k) \mathbf{A}(k) \tau) \boldsymbol{\xi}(k) \\
        &= (\mathbf{A}(k) + (\mathbf{A}(k) \mathbf{F}_b(k) + \dot{\mathbf{A}}(k)) \tau) \boldsymbol{\xi}(k) \\
        &\approx (\mathbf{A}(k+1) + \mathbf{A}(k) \mathbf{F}_b(k) \tau) \boldsymbol{\xi}(k) \\
        &\approx (\mathbf{A}(k+1) + (\mathbf{A}(k) + \dot{\mathbf{A}(k)} \tau) \mathbf{F}_b(k) \tau) \boldsymbol{\xi}(k) \\
        &\approx (\mathbf{A}(k+1) + \mathbf{A}(k+1) \mathbf{F}_b(k) \tau) \boldsymbol{\xi}(k) \\
        &\approx \mathbf{A}(k+1) (\mathbf{I} + \mathbf{F}_b(k) \tau) \boldsymbol{\xi}(k) \\
        &\approx \mathbf{A}(k+1) \boldsymbol{\xi}_b(k+1).
    \end{split}
    \end{equation}
\end{proof}
\end{theorem}

In practical applications, although discrete propagation introduces approximation errors, these do not significantly affect the equivalence of covariance propagation, and the above analysis is still valid in inertial-based navigation systems.
\section{Transformation Between EKF and InEKF Error States}\label{appx:transformation_relationship}
This appendix presents the relationship between the classical EKF error states and the error representations used in the L-InEKF and R-InEKF. The left-invariant Jacobian $\mathbf{J}_l$ is given by
\begin{equation}
    \mathbf{J}_l = \left[
    \begin{array}{ccc}
    -\hat{\mathbf{C}}_e^b & \mathbf{0}_3 & \mathbf{0}_3 \\
    \mathbf{0}_3 & -\hat{\mathbf{C}}_e^b & -\hat{\mathbf{C}}_e^b \boldsymbol{\Omega}_{ie}^e \\
    \mathbf{0}_3 & \mathbf{0}_3 & -\hat{\mathbf{C}}_e^b
    \end{array}
    \right],
\end{equation}
while the right-invariant Jacobian $\mathbf{J}_r$ is expressed as
\begin{equation}
    \mathbf{J}_r = \left[
    \begin{array}{ccc}
    -\mathbf{I}_3 & \mathbf{0}_3 & \mathbf{0}_3 \\
    -\hat{\overline{\mathbf{v}}}^e \times & -\mathbf{I}_3 & -\boldsymbol{\Omega}_{ie}^e \\
    -\hat{\mathbf{r}}^e \times & \mathbf{0}_3 & -\mathbf{I}_3
    \end{array}
    \right].
\end{equation}

The explicit form of the adjoint matrix $\mathbf{Ad}_{\hat{\boldsymbol{\chi}}}$ can be derived from the Jacobians $\mathbf{J}_l$ and $\mathbf{J}_r$ as follows
\begin{equation}
\begin{aligned}
    \boldsymbol{\xi}_r &= \mathbf{J}_r \, \delta \mathbf{x} \\
    &= \mathbf{J}_r \mathbf{J}_l^{-1} \boldsymbol{\xi}_l \\
    &= \left[
    \begin{array}{ccc}
        \hat{\mathbf{C}}_b^e & \mathbf{0}_3 & \mathbf{0}_3 \\
        \hat{\overline{\mathbf{v}}}^e \times \hat{\mathbf{C}}_b^e & \hat{\mathbf{C}}_b^e & \mathbf{0}_3 \\
        \hat{\mathbf{r}}^e \times \hat{\mathbf{C}}_b^e & \mathbf{0}_3 & \hat{\mathbf{C}}_b^e
    \end{array}
    \right] \boldsymbol{\xi}_l,
\end{aligned}
\end{equation}
where the adjoint representation is
\begin{equation}
    \mathbf{Ad}_{\hat{\boldsymbol{\chi}}} =
    \left[
    \begin{array}{ccc}
        \hat{\mathbf{C}}_b^e & \mathbf{0}_3 & \mathbf{0}_3 \\
        \hat{\overline{\mathbf{v}}}^e \times \hat{\mathbf{C}}_b^e & \hat{\mathbf{C}}_b^e & \mathbf{0}_3 \\
        \hat{\mathbf{r}}^e \times \hat{\mathbf{C}}_b^e & \mathbf{0}_3 & \hat{\mathbf{C}}_b^e
    \end{array}
    \right].
\end{equation}
According to Section~\ref{sec:error_state_relation}, the system matrices $\mathbf{F}_l$, $\mathbf{F}_r$, and $\mathbf{F}_{ekf}$, as well as the noise input matrices $\mathbf{G}_l$, $\mathbf{G}_r$, and $\mathbf{G}_{ekf}$, satisfy the following constraint relations constructed from $\mathbf{J}_l$, $\mathbf{J}_r$, and $\mathbf{Ad}_{\hat{\mathbf{x}}}$
\begin{equation}
    \begin{aligned}
        \mathbf{F}_l \mathbf{J}_l &= \dot{\mathbf{J}}_l + \mathbf{J}_l \mathbf{F}_{ekf}, \\
        \mathbf{F}_r \mathbf{J}_r &= \dot{\mathbf{J}}_r + \mathbf{J}_r \mathbf{F}_{ekf}, \\
        \mathbf{F}_r \mathbf{Ad}_{\hat{\mathbf{x}}} &= \dot{\mathbf{Ad}}_{\hat{\boldsymbol{\chi}}} + \mathbf{Ad}_{\hat{\mathbf{x}}} \mathbf{F}_l,
    \end{aligned}
\end{equation}
\begin{equation}
    \begin{aligned}
        \mathbf{G}_l &= \mathbf{J}_l \mathbf{G}_{ekf}, \\
        \mathbf{G}_r &= \mathbf{J}_r \mathbf{G}_{ekf}, \\
        \mathbf{G}_r &= \mathbf{Ad}_{\hat{\mathbf{x}}} \mathbf{G}_l.
    \end{aligned}
\end{equation}

\section{Covariance Transformation Matrices Among EKF, L-InEKF, and R-InEKF}\label{appx:trans_mat}
The explicit forms of the covariance transformation matrices $\mathbf{T}$ among the EKF, L-InEKF, and R-InEKF are given as follows
\begin{equation}
    \begin{aligned}
        &\mathbf{T}_{\mathrm{ekf} \to l} \\
        &= \mathbf{J}_l^{-1}(\hat{\mathbf{x}}^+) \mathbf{J}_l(\hat{\mathbf{x}}^-) \\
        &= \left[
        \begin{array}{ccc}
             \hat{\mathbf{C}}_b^{e+} \hat{\mathbf{C}}_e^{b-} & \mathbf{0}_3 & \mathbf{0}_3 \\
             \mathbf{0}_3 & \hat{\mathbf{C}}_b^{e+} \hat{\mathbf{C}}_e^{b-} & \hat{\mathbf{C}}_b^{e+} \hat{\mathbf{C}}_e^{b-} \boldsymbol{\Omega}_{ie}^e - \boldsymbol{\Omega}_{ie}^e \hat{\mathbf{C}}_b^{e+} \hat{\mathbf{C}}_e^{b-} \\
             \mathbf{0}_3 & \mathbf{0}_3 & \hat{\mathbf{C}}_b^{e+} \hat{\mathbf{C}}_e^{b-}
        \end{array}
        \right],
    \end{aligned}
\end{equation}

\begin{equation}
    \begin{aligned}
        \mathbf{T}_{l \to \mathrm{ekf}} &= \mathbf{J}_l(\hat{\mathbf{x}}^+) \mathbf{J}_l^{-1}(\hat{\mathbf{x}}^-) \\
        &= \left[
        \begin{array}{ccc}
             \hat{\mathbf{C}}_e^{b+} \hat{\mathbf{C}}_b^{e-} & \mathbf{0}_3 & \mathbf{0}_3 \\
             \mathbf{0}_3 & \hat{\mathbf{C}}_e^{b+} \hat{\mathbf{C}}_b^{e-} & \mathbf{0}_3 \\
             \mathbf{0}_3 & \mathbf{0}_3 & \hat{\mathbf{C}}_e^{b+} \hat{\mathbf{C}}_b^{e-}
        \end{array}
        \right],
    \end{aligned}
\end{equation}
\begin{equation}
    \begin{aligned}
        &\mathbf{T}_{ekf \to r} \\
        &= \mathbf{J}_r^{-1}(\hat{\mathbf{x}}^+) \mathbf{J}_r(\hat{\mathbf{x}}^-) \\
        &= \left [ 
        \begin{array}{ccc}
             \mathbf{I}_3 & \mathbf{0}_3 & \mathbf{0}_3 \\
             -\hat{\overline{\mathbf{v}}}^{e+} \times + \boldsymbol{\Omega}_{ie}^e \hat{\mathbf{r}}^{e+} \times + \hat{\overline{\mathbf{v}}}^{e-} \times - \boldsymbol{\Omega}_{ie}^e \hat{\mathbf{r}}^{e-} \times& \mathbf{I}_3 & \mathbf{0}_3 \\
             -\hat{\mathbf{r}}^{e+} \times + \hat{\mathbf{r}}^{e-} \times & \mathbf{0}_3 & \mathbf{I}_3
        \end{array}
        \right],
    \end{aligned}
\end{equation}
\begin{equation}
    \begin{aligned}
        \mathbf{T}_{r \to ekf} &= \mathbf{J}_r(\hat{\mathbf{x}}^+) \mathbf{J}_r^{-1}(\hat{\mathbf{x}}^-) \\
        &= \left [ 
        \begin{array}{ccc}
             \mathbf{I}_3 & \mathbf{0}_3 & \mathbf{0}_3 \\
             \hat{\overline{\mathbf{v}}}^{e+} \times - \hat{\overline{\mathbf{v}}}^{e-} \times & \mathbf{I}_3 & \mathbf{0}_3 \\
             \hat{\mathbf{r}}^{e+} \times - \hat{\mathbf{r}}^{e-} \times & \mathbf{0}_3 & \mathbf{I}_3
        \end{array}
        \right],
    \end{aligned}
\end{equation}

\begin{equation}
    \begin{aligned}
        &\mathbf{T}_{l \to r} \\
        &= \mathbf{Ad}_{\hat{\boldsymbol{\chi}}^+}^{-1} \mathbf{Ad}_{\hat{\boldsymbol{\chi}}^-} \\
        &= \left [ 
        \begin{array}{ccc}
             \hat{\mathbf{C}}_e^{b+} \hat{\mathbf{C}}_b^{e-} & \mathbf{0}_3 & \mathbf{0}_3 \\
             (-\hat{\mathbf{C}}_e^{b+} \hat{\overline{\mathbf{v}}}^{e+} \times \hat{\mathbf{C}}_b^{e-} \\+ \hat{\mathbf{C}}_e^{b+} \hat{\overline{\mathbf{v}}}^{e-} \times \hat{\mathbf{C}}_b^{e-})& \hat{\mathbf{C}}_e^{b+} \hat{\mathbf{C}}_b^{e-} & \mathbf{0}_3 \\
             (-\hat{\mathbf{C}}_e^{b+} \hat{\mathbf{r}}^{e+} \times \hat{\mathbf{C}}_b^{e-} \\+ \hat{\mathbf{C}}_e^{b+} \hat{\mathbf{r}}^{e-} \times \hat{\mathbf{C}}_b^{e-}) & \mathbf{0}_3 & \hat{\mathbf{C}}_e^{b+} \hat{\mathbf{C}}_b^{e-}
        \end{array}
        \right],
    \end{aligned}
\end{equation}

\begin{equation}
    \begin{aligned}
        &\mathbf{T}_{r \to l} \\
        &= \mathbf{Ad}_{\hat{\boldsymbol{\chi}}^+} \mathbf{Ad}_{\hat{\boldsymbol{\chi}}^-}^{-1} \\
        &= \left [ 
        \begin{array}{ccc}
             \hat{\mathbf{C}}_b^{e+} \hat{\mathbf{C}}_e^{b-} & \mathbf{0}_3 & \mathbf{0}_3 \\
             (\hat{\overline{\mathbf{v}}}^{e+} \times \hat{\mathbf{C}}_b^{e+} \hat{\mathbf{C}}_e^{b-} \\- \hat{\mathbf{C}}_b^{e+} \hat{\mathbf{C}}_e^{b-} \hat{\overline{\mathbf{v}}}^{e-} \times) & \hat{\mathbf{C}}_b^{e+} \hat{\mathbf{C}}_e^{b-} & \mathbf{0}_3 \\
             (\hat{\mathbf{r}}^{e+} \times \hat{\mathbf{C}}_b^{e+} \hat{\mathbf{C}}_e^{b-} \\- \hat{\mathbf{C}}_b^{e+} \hat{\mathbf{C}}_e^{b-} \hat{\mathbf{r}}^{e-} \times) & \mathbf{0}_3 & \hat{\mathbf{C}}_b^{e+} \hat{\mathbf{C}}_e^{b-}
        \end{array}
        \right].
    \end{aligned}
\end{equation}
Further derivation of the above equations reveals that they are all related to the estimated error states.



\bibliographystyle{IEEEtran}
\bibliography{ref/refs}

\end{document}